\documentclass[11pt]{article}
\usepackage{odonnell-csp-ref}

\title{How to refute a random CSP}
\author{Sarah R. Allen\thanks{Department of Computer Science, Carnegie Mellon. \texttt{\{srallen,odonnell,dwitmer\}@cs.cmu.edu}. Supported by NSF grants CCF-0747250 and CCF-1116594.  Some of this work performed while the second-named author was at the Bo\u{g}azi\c{c}i University Computer Engineering Department, supported by Marie Curie International Incoming Fellowship project number 626373.  The first and third named authors were partially supported by the National Science Foundation Graduate Research Fellowship Program under Grant No. DGE-1252522.} \and Ryan O'Donnell${}^*$ \and David Witmer${}^*$}

\begin{document}
\maketitle

\begin{abstract}
    Let $P$ be a nontrivial 
    $k$-ary predicate over a finite alphabet. Consider a random CSP$(P)$ instance~$\calI$ over $n$ variables with $m$ constraints, each being $P$ applied to $k$ random literals. 
    When $m \gg n$ the instance $\calI$ will be unsatisfiable with high probability, and the natural associated algorithmic task is to find a \emph{refutation} of~$\calI$ --- i.e., a certificate of unsatisfiability.  When $P$ is the $3$-ary Boolean OR predicate, this is the well studied problem of refuting random $3$-SAT formulas; in this case, an efficient algorithm is known only when $m \gg n^{3/2}$. Understanding the density required for average-case refutation of other predicates is of importance for various areas of complexity, including cryptography, proof complexity, and learning theory. The main previously-known result is that for a general Boolean $k$-ary predicate~$P$, having $m \gg n^{\lceil k/2 \rceil}$ random constraints suffices for efficient refutation.
    
    In this work we give a general criterion for arbitrary $k$-ary predicates, one that often yields efficient refutation algorithms at much lower densities.  Specifically, if~$P$ fails to support a $t$-wise independent (uniform) probability distribution ($2 \leq t \leq k$), then there is an efficient algorithm that refutes random CSP$(P)$ instances $\calI$ with high probability, provided $m \gg n^{t/2}$.  Indeed, our algorithm will ``somewhat strongly'' refute~$\calI$, certifying $\Opt(\calI) \leq 1 - \Omega_k(1)$; if $t = k$ then we furthermore get the strongest possible refutation, certifying $\Opt(\calI) \leq \E[P] + o(1)$.  This last result is new even in the context of random $k$-SAT.
    
    Regarding the optimality of our $m \gg n^{t/2}$ density requirement, prior work on SDP hierarchies has given some evidence that efficient refutation of random CSP$(P)$ may be impossible when $m \ll n^{t/2}$. Thus there is an indication our algorithm's dependence on~$m$ is optimal for every~$P$, at least in the context of SDP hierarchies. Along these lines, we show that our refutation algorithm can be carried out by the $O(1)$-round SOS SDP hierarchy.
    
    Finally, as an application of our result, we falsify the ``SRCSP assumptions'' used to show various hardness-of-learning results in the recent (STOC 2014) work of Daniely, Linial, and Shalev--Shwartz.
\end{abstract}

\thispagestyle{empty}
\setcounter{page}{0}
\newpage


\newcommand{\dens}{\alpha}

\section{On refutation of random CSPs}

Constraint satisfaction problems (CSPs) play a major role in computer science.  There is a vast theory~\cite{BJK05} of how algebraic properties of a CSP predicate affect its worst-case satisfiability complexity; there is a similarly vast theory~\cite{Rag09} of worst-case approximability of CSPs.  Finally, there is a rich range of research --- from the fields of computer science, mathematics, and physics --- on the \emph{average-case} complexity of random CSPs; see~\cite{Ach09} for a survey just of random $k$-SAT.  This paper is concerned with random CSPs, and in particular the problem of efficiently \emph{refuting} satisfiability for random instances.  This is a well-studied algorithmic task with connections to, e.g., proof complexity~\cite{BSB02}, inapproximability~\cite{Fei02}, SAT-solvers~\cite{SAT14}, cryptography~\cite{ABW10}, learning theory~\cite{daniely-hardness-of-learning}, statistical physics~\cite{CLP02}, and complexity theory~\cite{BKS13}.

Historically, random CSPs are probably best studied in the case of $k$-SAT, $k \geq 3$. The model here involves choosing a CNF formula~$\calI$ over $n$ variables by drawing $m$ clauses (ORs of $k$ literals) independently and uniformly at random. (The precise details of the random model are inessential; see Section~\ref{sec:csp-prelims} for more information.)
This is one of the best known efficient ways of generating hard-seeming instances of $\NP$-complete and $\coNP$-complete problems.  The computational hardness depends crucially on the \emph{density}, $\dens = m/n$. For each~$k$ there is (conjecturally) a constant \emph{critical density} $\alpha_k$ such that $\calI$ is satisfiable with high probability when $\dens < \dens_k$, and $\calI$ is unsatisfiable with high probability when $\dens > \dens_k$. (Here and throughout, ``with high probability (whp)'' means with probability $1-o(1)$ as $n \to \infty$.)
This phenomenon occurs for all nontrivial random CSPs; in the case of $k$-SAT it's been rigorously proven~\cite{DSS15} for sufficiently large $k$.

There is a natural algorithmic task associated with each of the two regimes. When $\dens < \dens_k$ one wants to find a satisfying assignment for~$\calI$. When $\dens > \dens_k$ one wants to \emph{refute}~$\calI$; i.e., find a \emph{certificate} of unsatisfiability. Most heuristic SAT-solvers use DPLL-based algorithms; on unsatisfiable instances, they produce certificates that can be viewed as refutations within the Resolution proof system. More generally, a \emph{refutation algorithm} for density~$\dens$  is any algorithm that: a)~outputs ``unsatisfiable'' or ``fail''; b)~never incorrectly outputs ``unsatisfiable''; c)~outputs ``fail'' with low probability (i.e., probability~$o(1)$).\footnote{We caution the reader that in this paper we do not consider the related, but distinct, scenario of distinguishing \emph{planted} random instances from truly random ones.}
Empirical work suggests that as $\dens$ increases towards $\dens_k$, finding satisfying assignments becomes more difficult; and conversely, as $\dens$ increases beyond $\dens_k$, finding certificates of unsatisfiability gradually becomes easier.

A seminal paper of Chv{\'a}tal and Szemer{\'e}di~\cite{CS88} showed that for any sufficiently large integer~$c$ (depending on~$k$), a random $k$-SAT instance with $m = cn$ requires Resolution refutations of size $2^{\Omega(n)}$ (whp).  On the other hand, Fu~\cite{Fu96} showed that polynomial-size Resolution refutations exist (whp) once $m \geq O(n^{k-1})$;  Beame et~al.~\cite{BKPS98} subsequently showed that such proofs could be found efficiently.\footnote{In this paper we use the following not-fully-standard terminology: A statement of the form ``If $f(n) \geq O(g(n))$ then~$X$'' means that there exists a certain function~$h(n)$, with $h(n)$ being $O(g(n))$, such that the statement ``If $f(n) \geq h(n)$ then~$X$'' is true. We also use $\wt{O}(f(n))$ to denote $O(f(n) \cdot \polylog(f(n))$, and $O_k(f(n))$ to denote that the hidden constant has a dependence on~$k$ (most often of the form $2^{O(k)}$).}
A breakthrough came in 2001, when Goerdt and Krivelevich~\cite{GK01} abandoned combinatorial refutations for spectral ones, showing that random $k$-SAT instances can be efficiently refuted when $m \geq \wt{O}(n^{\lceil k/2 \rceil})$.  Soon thereafter, Friedman and Goerdt~\cite{FG01} (see also~\cite{FGK05}) showed that for $3$-SAT, efficient spectral refutations exist once $m \geq n^{3/2 + \eps}$ (for any $\eps > 0$).  These densities for $k$-SAT --- around~$n^{3/2}$ for $3$-SAT and $n^{\lceil k/2 \rceil}$ in general --- have not been fundamentally improved upon in the last~14 years.\footnote{Actually, it is claimed in~\cite{GJ02} that one can obtain $n^{k/2+\eps}$ for odd~$k$ ``along the lines of~\cite{FG01}''.  On one hand, this is true, as we'll see in this paper.  On the other hand, no proof was provided in~\cite{GJ02}, and we have not found the claim repeated in any paper subsequent to 2003.}  (See Table~\ref{table:refutations} for a more detailed history of results in this area.)  Improving the~$n^{3/2}$ bound for $3$-SAT is widely regarded as a major open problem~\cite{ABW10}, with conjectures regarding its possibility going both ways~\cite{BSB02,NIPS2013_4905}.  See also the intriguing work of Feige, Kim, and Ofek~\cite{FKO06} showing that polynomial-size $3$-SAT refutations \emph{exist} whp once $m \geq O(n^{1.4})$.

\paragraph{Strong refutation.}  In a notable paper from 2002, Feige~\cite{Fei02} made a fruitful connection between the hardness of refuting random $3$-SAT instances and the inapproximability of certain optimization problems that are challenging to analyze by other means.
This refers to certifying not only that a random instance $\calI$ is unsatisfiable, but furthermore that $\Opt(\calI) \leq 1-\delta$ for some constant $\delta > 0$. Here $\Opt(\calI)$ denotes the maximum \emph{fraction} of simultaneously satisfiable constraints in~$\calI$.
Feige specifically introduced the following ``R3SAT Hypothesis'': \emph{For all small $\delta > 0$ and large $c \in \N$, there is no polynomial-time $\delta$-refutation algorithm for random $3$-SAT with $m = cn$.}
To stress-test Feige's R3SAT Hypothesis, one may ask if the aforementioned refutation algorithms for $k$-SAT can be improved to $\delta$-refutation algorithms.  Coja-Oghlan et~al.~\cite{COGL04} showed that they can be in the case of $k = 3,4$.  Indeed, they gave algorithms for what is called \emph{strong refutation} in these cases.  Here strongly refuting $k$-SAT refers to certifying that $\Opt(\calI) \leq 1 - 2^{-k} + o(1)$ (note that $\Opt(\calI) \approx 1 - 2^{-k}$ whp assuming $m \geq O(n)$).

\paragraph{Beyond $k$-SAT.}  As in the algebraic and approximation theories of CSP, it's of significant interest to consider random instances of the CSP$(P)$ problem for general predicates $P \co \Z_q^k \to \{0,1\}$, besides just Boolean~OR.  (Though Boolean predicates are more familiar, larger domains are of interest, e.g., for $q$-colorability of $k$-uniform hypergraphs.) Specifically, we are interested in the question of how properties of~$P$ affect the number of constraints needed for efficient refutation of random CSP$(P)$ instances.
This precise question is very relevant for work in cryptography based on the candidate OWFs and PRGs of Goldreich~\cite{Gol00}; see also~\cite{ABW10} and the survey of Applebaum~\cite{App13}.  It has also proven essential for the recent exciting approach to hardness of learning due to Daniely, Linial, and Shalev-Shwartz~\cite{NIPS2013_4905,daniely-hardness-of-learning, 2014arXiv1404.3378D}.  We discuss this learning connection and our results on it in more detail in Section~\ref{sec:learning}.

The special case of random $3$-XOR has proved particularly important: it is related to $3$-SAT refutation through Feige's ``3XOR Principle'' (see~\cite{Fei02,FO05,FKO06}); it's the basis for cryptographic schemes due to Alekhnovich~\cite{Ale03} (and is related to the ``Learning Parities with Noise'' problem); it's used in the best known lower bounds for the SOS SDP hierarchy~\cite{Gri01,Sch08}, which we discuss further in Section~\ref{sec:sos}; and, Barak and Moitra~\cite{BM15} have shown it to be equivalent to a certain ``tensor prediction problem'' in learning theory.

Prior to this work, very little was known about how the predicate~$P$ affects the complexity of refuting random CSP$(P)$ instances.  The main known result, following from the work Coja-Oghlan, Cooper, and Frieze~\cite{COCF10}, was the following: For any Boolean $k$-ary predicate~$P$, one can efficiently strongly refute random CSP$(P)$ instances~$\calI$ (i.e., certify $\Opt(\calI) \leq \E[P] + o(1)$)  provided the number of constraints~$m$ satisfies $m \geq \wt{O}(n^{\lceil k/2 \rceil})$.  In the case of $k$-XOR, the very recent work of Barak and Moitra~\cite{BM15} showed how to improve this bound to $m \geq \wt{O}(n^{k/2})$.\footnote{The present authors also obtained this result around the same time, but we credit the result to~\cite{BM15} as they published earlier.  With their permission we repeat the proof herein, partly because we need to prove a slightly more general variant.}

\newcolumntype{Z}{>{\centering\let\newline\\\arraybackslash\hspace{0pt}}X}
\nopagebreak
\begin{table}           \label{table:refutations}
\begin{tabularx}{\textwidth}{|Z|Z|Z|Z|Z|}
\hline
 CSP & Poly-size refutations \newline
       whp once $m \geq \cdots$      & Strength              & Efficient/
                                                      \newline Existential  & Reference              \\ \hline \hline
$k$-SAT  & $O(n^{k-1})$              &        Refutation      & Existential & \cite{Fu96}            \\ \hline
$k$-SAT  & $O(n^{k-1}/{\scriptscriptstyle \log^{k-2}(n)})$ &
                                              Refutation      & Efficient   & \cite{BKPS98,BKPS02}   \\ \hline
$k$-SAT  & $\wt{O}(n^{\lceil k/2 \rceil})$&   Refutation      & Efficient   & \cite{GK01,FGK05}      \\ \hline
$3$-SAT  & $O(n^{3/2 + \eps})$       &        Refutation      & Efficient   & \cite{FG01,FGK05}      \\ \hline
$k$-SAT  & $O(n^{k/2 + \eps})$       &        Refutation      & Efficient   & Claimed possible in \cite{GJ02,GJ03} \\ \hline
Exactly-$k_1$-out-of-$k$-SAT & $\wt{O}(n)$ &  Refutation      & Efficient   & \cite{BSB02}           \\ \hline
$2$-out-of-$5$-SAT & $O(n^{3/2 + \eps})$ &    Refutation      & Efficient   & \cite{GJ02,GJ03}       \\ \hline
NAE-$3$-SAT & $O(n)$                 & $\Omega(1)$-Refutation & Efficient   & \cite{KLP96,GJ02,GJ03} \\ \hline
$k$-SAT  & $O(n^{\lceil k/2 \rceil})$&        Refutation      & Efficient   & \cite{CGLS04,FO05}    \\ \hline
$3$-SAT  & $\wt{O}(n^{3/2})$         &        Refutation      & Efficient   & \cite{GL03}            \\ \hline
$3$-SAT  & $\wt{O}(n^{3/2})$         & Strong                 & Efficient   & \cite{COGL04,COGL07}   \\ \hline
$4$-SAT  & $O(n^2)$                  & Strong                 & Efficient   & \cite{COGL04,COGL07}   \\ \hline
$3$-SAT  & $O(n^{3/2})$              &        Refutation      & Efficient   & \cite{FO04}            \\ \hline
$3$-SAT  & $O(n^{1.4})$              &        Refutation      & Existential & \cite{FKO06}           \\ \hline
$3$-XOR  & $O(n^{3/2})$              & $1/n^{\Omega(1)}$-refutation & Efficient   & Implicit in \cite{FKO06}    \\ \hline
$3$-SAT  & $\Omega(n^{3/2})$         &        Refutation      & Efficient   & Claimed in \cite{FKO06}\\ \hline
Boolean $k$-CSP &$O(n^{\lceil k/2 \rceil})$& Strong           & Efficient   & \cite{COCF10}          \\ \hline
$k$-XOR  & $\wt{O}(n^{k/2})$         &       Strong           & Efficient   & \cite{BM15} (also
                                                                               herein)      \\ \hline
Any $k$-CSP     &  $\wt{O}(n^{k/2})$ & Quasirandom \newline
                                       ($\implies$ strong)    & Efficient   &          This paper    \\ \hline
Any $k$-CSP \newline
not supporting \newline
$t$-wise indep. &\newline $\wt{O}(n^{t/2})$&
                     \newline $\Omega_k(1)$-refutation & \newline Efficient   & \newline This paper    \\ \hline
\end{tabularx}
\caption{Up to logarithmic factors on $m$, our work subsumes all previously known results.}
\end{table}


\section{Our results and techniques}\label{sec:results}

Here we describe our main results and techniques at a high level.  Precise theorem statements appear later in the work and the definitions of the terminology we use is given in Section~\ref{sec:prelims}.  We also mention that in Section~\ref{sec:larger-alphabets} we will generalize all of our results to the case of larger alphabets; but we'll just discuss Boolean predicates $P : \{0,1\}^k \to \{0,1\}$  for simplicity.

Our main result gives a (possibly sharp) bound on the number of constraints needed to refute random $\CSP(P)$ instances.  Before getting to it, we first describe some more concrete results that go into the main proof.  All of our results rely on a strong refutation algorithm for $k$-XOR (actually, a slight generalization thereof).  For $m \geq \wt{O}(n^{\lceil k/2 \rceil})$, such a result follows from~\cite{COCF10}; however, the exponent $\lceil k/2 \rceil$ can be improved to $k/2$.  We give a demonstration of this fact herein; as mentioned earlier, it was published very recently by Barak and Moitra \cite[Corollary 5.5 and Extensions]{BM15}.
\begin{theorem}
\label{thm:xor-intro}
There is an efficient algorithm that (whp) strongly refutes random $k$-XOR instances with at least $\wt{O}(n^{k/2})$ constraints; i.e., it certifies $\Opt(\calI) \leq \frac12 + o(1)$.
\end{theorem}
The proof of Theorem~\ref{thm:xor-intro} follows ideas from~\cite{COGL07} and earlier works on ``discrepancy'' of random $k$-SAT instances.  The case of even~$k$ is notably easier, and we present two ``folklore'' arguments for it.  The case of odd~$k$ is trickier. Roughly speaking we view the instance as a homogeneous degree-$k$ multilinear polynomial, which we want to certify takes on only small values on inputs in $\{-1,1\}^n$. Considering separately the contributions based on the ``last'' of the~$k$ variables in each constraint, and then using Cauchy--Schwarz, it suffices to bound the norm of a carefully designed quadratic form of dimension~$n^{k-1}$, indexed by tuples of $k-1$ variables.  This is done using the trace method~\cite{Wig55,FK81}.  Similar techniques, including the use of the trace method, date back to the 2001 Friedgman--Goerdt work~\cite{FG01} refuting random $3$-SAT given $m = n^{3/2 + \eps}$ constraints.

With Theorem~\ref{thm:xor-intro} in hand, the next step is certifying \emph{quasirandomness} of random $k$-ary CSP instances having $m \geq \wt{O}(n^{k/2})$ constraints.  Roughly speaking we say that a CSP instance is quasirandom if, for every assignment $x \in \{0,1\}^n$, the $m$ induced $k$-tuples of literal values are close to being uniformly distributed over $\{0,1\}^k$.  (Note that this is only a property of the instances' constraint scopes/negations, and has nothing to do with~$P$.) Since the ``Vazirani XOR Lemma'' implies that a distribution on $\{-1,1\}^k$ is uniform if and only if all its~$2^k$ XORs are have bias~$\frac12$, we are able to leverage Theorem~\ref{thm:xor-intro} to prove:
\begin{theorem}
\label{thm:quasi-intro}
There is an efficient algorithm that (whp) certifies that a random instance of $\CSP(P)$ is quasirandom, provided the number of constraints is at least $\wt{O}(n^{k/2})$.
\end{theorem}

If an instance is quasirandom, then no solution can be much better than a randomly chosen one.  Thus by certifying quasirandomness we are able to strongly refute random instances of any CSP$(P)$:
\begin{theorem}
\label{thm:strong-ref-intro}
For any $k$-ary predicate~$P$, there is an efficient algorithm that (whp) strongly refutes random $\CSP(P)$ instances when the number of constraints is at least $\wt{O}(n^{k/2})$.
\end{theorem}
In particular, this theorem improves upon~\cite{COCF10} by a factor of~$\sqrt{n}$ whenever $k$ is odd; this savings is new even in the well-studied case of $k$-SAT.

The above result does not make use of any properties of the predicate~$P$ other than its arity,~$k$. We now come to our main result, which shows that for many interesting~$P$, random $\CSP(P)$  instances can be refuted with many fewer constraints than $n^{k/2}$.  In the following theorem, the phrase ``$t$-wise uniform'' (often imprecisely called ``$t$-wise independent'') refers to a distribution on $\bits^k$ in which all marginal distributions on $t$ out of $k$ coordinates are uniform.
\begin{theorem} (Main.)
    \label{thm:non-t-wise-intro}
    Let $P$ be a $k$-ary predicate such that there is no $t$-wise uniform distribution supported on its satisfying assignments, $t \geq 2$.  Then there is an efficient algorithm that (whp) $\Omega_k(1)$-refutes random instances of $\CSP(P)$ with at least $\wt{O}(n^{t/2})$ constraints.
\end{theorem}
We remark that property of a predicate $P$ supporting a \pw uniform distribution has  played an important role in approximability theory for CSPs, ever since Austrin and Mossel~\cite{AM09} showed that such predicates are hereditarily approximation-resistant under the UGC.  Also, note that the largest $t$ for which a predicate $P$ supports a $t$-wise uniform distribution determines the minimum number of constraints required by our algorithm.  This value is closely related to the notion of distribution complexity studied by Feldman, Perkins, and Vempala~\cite{FPV14,FPV15} in the context of planted random CSPs.  Informally, the distribution complexity of a planted CSP is the largest $t$ for which the distribution over constraint inputs $\{-1,1\}^k$ induced by the planted assignment is $t$-wise uniform.  Despite this similarity, the algorithmic techniques used by Feldman, Perkins, and Vempala in the planted case~\cite{FPV14} do not seem to directly apply to refutation.

The idea behind the proof of Theorem~\ref{thm:non-t-wise-intro} is that with $\wt{O}(n^{t/2})$ constraints we can use the algorithm of Theorem~\ref{thm:quasi-intro} to certify quasirandomness (closeness to uniformity) for all subsets of~$t$ out of~$k$ coordinates. Thus for every assignment $x \in \{0,1\}^n$, the induced distribution on constraint $k$-tuples is ($o(1)$-close to) $t$-wise uniform.  Since $P$ does not support a $t$-wise uniform distribution, this essentially shows that no~$x$ can induce a fully-satisfying distribution on constraint inputs.  To handle the $o(1)$-closeness caveat, we show that if $P$ does not support a $t$-wise uniform distribution, then it is $\delta$-far from supporting such a distribution, for $\delta = \Omega_k(1)$.  The algorithm can then in fact $(\delta-o(1))$-refute random CSP$(P)$ instances.

\begin{example}
    To briefly illustrate the result, consider the Exactly-$k$-out-of-$2k$-SAT CSP, studied previously in~\cite{BSB02,GJ03}.
    The associated predicate supports a $1$-wise uniform distribution, namely the uniform distribution over strings in $\{0,1\}^{2k}$ of Hamming weight~$k$.  However, it is not hard to show that it does not support any \pw uniform distribution. As a consequence, random instances of this CSP can be refuted with only $\wt{O}(n)$ constraints, independent of~$k$.
\end{example}

\subsection{An application from learning theory}\label{ssec:ltprelims}
Recently, an exciting approach to proving hardness-of-learning results has been developed by Daniely, Linial, and Shalev-Shwartz~\cite{NIPS2013_4905,daniely-hardness-of-learning, 2014arXiv1404.3378D, Dan15}.
The most general results appear in~\cite{daniely-hardness-of-learning}.  In this work, Daniely et al.\ prove computational hardness of several central learning theory problems, based on two assumptions concerning the hardness of random CSP refutation.  While the assumptions made in~\cite{2014arXiv1404.3378D,  Dan15} appear to be plausible, our work unfortunately shows that the more general assumptions made in~\cite{daniely-hardness-of-learning} are false.

Below we state the (admittedly strong) assumptions from~\cite{daniely-hardness-of-learning} (up to some very minor technical details which are discussed and treated in Section~\ref{sec:learning}).  We will need one piece of terminology: the \emph{$0$-variability} $\mathrm{VAR}_0(P)$ of a predicate $P : \{-1,1\}^k \to \{0,1\}$ is the least~$c$ such that there exists a restriction to some~$c$ input coordinates forcing~$P$ to be~$0$.  Essentially, the assumptions state that one can obtain hardness-of-refutation with an arbitrarily large polynomial number of constraints by using a family of predicates~$(P_k)$ that: a)~have unbounded $0$-variability; b)~support \pw uniformity. However, our work shows that supporting $t$-wise uniformity for unbounded~$t$ is also necessary.

\begin{named}{SRCSP Assumption 1} \label{assumption:intro-1}
(\!\!\cite{daniely-hardness-of-learning}.) For all $d \in \N$ there is a large enough $C$ such that the following holds:  If $P : \{-1,1\}^k \to \{0,1\}$ has  $\textrm{VAR}_0(P) \geq C$ and is hereditarily approximation resistant on satisfiable instances, then there is no polynomial-time algorithm refuting (whp) random instances of CSP$(P)$ with $m = n^d$ constraints.
\end{named}

\begin{named}{SRCSP Assumption 2}
\label{assumption:intro-2}
    (\!\!\cite{daniely-hardness-of-learning}, generalizing the ``RCSP Assumption'' of~\cite{BKS13} to superlinearly many constraints.)  For all $d \in \N$ there is a large enough $C$ such that the following holds:  If $P : \{-1,1\}^k \to \{0,1\}$ has  $\textrm{VAR}_0(P) \geq C$ and is $\delta$-close to supporting a \pw uniform distribution, then for all $\eps > 0$ there is no polynomial-time algorithm that $(\delta+\eps)$-refutes (whp) random instances of CSP$(P)$ with $m = n^d$ constraints.
\end{named}

In~\cite{daniely-hardness-of-learning} it is shown how to obtain three very notable hardness-of-learning results from these assumptions. However as stated, our work falsifies the SRCSP Assumptions.  Indeed, the assumptions are false even in the three specific cases used by~\cite{daniely-hardness-of-learning} to obtain hardness-of-learning results.  We now describe these cases.

\paragraph{Case 1.} The Huang predicates $(H_\kappa)$ are arity-$\Theta(\kappa^3)$ predicates introduced in~\cite{Huang:2013:ARS:2488608.2488666}; they are hereditarily approximation resistant on satisfiable instances and have $0$-variability $\Omega(\kappa)$.  In~\cite{daniely-hardness-of-learning} they are used in SRCSP Assumption~1 to deduce hardness of PAC-learning DNFs with $\omega(1)$ terms.  However:
\begin{theorem}
\label{thm:huang-intro}
For all $\kappa \geq 9$, the predicate $H_\kappa$ does not support a $4$-wise uniform distribution.
\end{theorem}
\noindent Thus by Theorem~\ref{thm:non-t-wise-intro} we can efficiently refute random instances of~CSP($H_\kappa$) with just $\wt{O}(n^2)$ constraints, independent of~$\kappa$.  This contradicts SRCSP Assumption~1.

\paragraph{Case 2.} The majority predicate $\Maj_k$ has $0$-variability  $\lceil k/2 \rceil$ and is shown in~\cite{daniely-hardness-of-learning} to be $\frac{1}{k+1}$-far from supporting a \pw uniform distribution.  In~\cite{daniely-hardness-of-learning} these predicates are used in SRCSP Assumption~2 to deduce hardness of agnsotically learning Boolean halfspaces to within any constant factor.  However:
\begin{theorem}
\label{thm:maj-intro}
For odd $k \geq 25$, the predicate $\Maj_k$ does not support a $4$-wise uniform distribution; in fact, it is $.1$-far from supporting a $4$-wise uniform distribution.
\end{theorem}
\noindent Theorem~\ref{thm:non-t-wise-intro} then implies we can efficiently $\delta$-refute random instances of CSP$(\Maj_k)$ with $\wt{O}(n^2)$ constraints, where $\delta = .1 \gg \frac{1}{k+1}$.  This contradicts SRCSP Assumption~2.

\paragraph{Case 3.} Finally, we also prove that SRCSP Assumption~1 is false for another family of predicates~$(T_k)$ used by~\cite{daniely-hardness-of-learning} to show hardness of PAC-learning intersections of $4$~Boolean halfspaces.\\

Our results described in these three cases all use linear programming duality. Specifically, in Lemma~\ref{lem:poly-iff-no-dist} we show that $P$ is $\delta$-far from supporting a $t$-wise uniform distribution if and only if there exists a $k$-variable multilinear polynomial~$Q$ that satisfies certain properties involving~$P$ and~$\delta$.  We then explicitly construct these dual polynomials for the Huang, Majority, and $T_k$ predicates.\\

We conclude this section by emphasizing the importance of the Daniely--Linial--Shalev-Shwartz hardness-of-learning program, despite the above results.  Indeed, subsequently to~\cite{daniely-hardness-of-learning}, Daniely and Shalev-Shwartz~\cite{2014arXiv1404.3378D} showed hardness of improperly learning $\textrm{DNF}$ formulas with~$\omega(\log n)$ terms under a much weaker assumption than SRCSP Assumption~$1$.  Specifically, their work only assumes that for all~$d$ there is a large enough~$k$ such that refuting random $k$-SAT instances is hard when there are $m = n^d$ constraints.  This assumption looks quite plausible to us, and may even be true with~$k$ not much larger than~$2d$.
Most recently, Daniely showed hardness of approximately agnostically learning halfspaces using the XOR predicate rather than majority~\cite{Dan15}.
This result shows that there is no efficient algorithm that agnostically learns halfspaces to within a constant approximation ratio under the assumption that refuting random $k$-XOR instances is hard when $m = n^{c\sqrt{k}\log k}$ for some $c > 0$.  
He also shows hardness of learning halfspaces to within an approximation factor of $2^{\log^{1-\nu}n}$ for any $\nu>0$ assuming that there exists some constant $c >0$ such that for all $s$, refuting random $k$-XOR instances with $k = \log^sn$ is hard when $m = n^{ck}$.

\subsection{Evidence for the optimality of our results}
It's natural to ask whether the $n^{t/2}$ dependence in our main Theorem~\ref{thm:non-t-wise-intro} can be improved.  As we don't expect to prove unconditional hardness results, we instead merely seek good evidence that refuting $(t-1)$-wise supporting predicates~$P$ is hard when $m \ll n^{t/2}$.  A natural form of evidence would be showing that various strong classes of polynomial-time refutation algorithms fail when $m \ll n^{t/2}$.  To make sense of this we need to talk about the form of such algorithms; i.e., propositional proof systems.  

Recently, there has been significant study of the ``SOS'' (Sum-Of-Squares) proof system, introduced by Grigoriev and Vorobjov~\cite{GV01}; see, e.g.,~\cite{OZ13,BS14} for discussion.  It has the following virtues: a)~it is very powerful, being able to efficiently simulate other proof systems (e.g., Resolution, Lov\'{a}sz--Schrijver); b)~it is automatizable~\cite{Las00,Par00}, meaning that $n$-variable ``degree-$d$ SOS proofs'' can be found in $n^{O(d)}$ time whenever they exist; c)~we do know some examples of \emph{lower bounds} for degree-$d$ SOS proofs.  In Section~\ref{sec:sos} of this paper we show the following:
\begin{theorem}                                     \label{thm:sos-able}
    All of our refutation algorithms for $k$-ary predicates can be extended to produce degree-$2k$ SOS proofs.  
\end{theorem}

We now return to the question of evidence for the optimality of constraint density used in our results.  Dating back to Franco--Paull~\cite{FP83} and Chv{\'a}tal--Szemer{\'e}di~\cite{CS88}, there has been a long line of work in proof complexity showing lower bounds for refuting random $3$-SAT instances, especially in the Resolution proof system.  This culminated in the work of Ben-Sasson and Wigderson~\cite{BSW99}, which showed that for random $3$-SAT (and $3$-XOR) with $m = O(n^{3/2 - \eps})$, Resolution refutations require size $2^{n^{\Omega(\eps)}}$ (whp). 
More recently, Schoenebeck~\cite{Sch08} showed (using the expansion analysis of~\cite{BSW99}) that random $k$-XOR and $k$-SAT instances with $m \leq n^{k/2 - \eps}$ require SOS proofs of degree $n^{\Omega(\eps)}$, and hence take $2^{n^{\Omega(\eps)}}$ time to refute by the ``SOS Method''. See~\cite{Tul09,Cha13} for related larger-alphabet followups.  These results show that the Barak--Moitra~$\wt{O}(n^{k/2})$ bound for refuting random $k$-XOR (which also works in $O(k)$-degree SOS) and our bound for random $k$-SAT are tight (up to subpolynomial factors) within the SOS framework.  Given the power of the SOS framework, this arguably constitutes some reasonable evidence that no polynomial-time algorithm can refute random $k$-SAT instances with $m \ll n^{k/2}$.

We now discuss our main theorem's $n^{t/2}$ bound for predicates $P$ not supporting $t$-wise uniform distributions.  Suppose $P$ is a predicate that does support a $(t-1)$-wise uniform distribution, where $t > 2$.   In the context of inapproximability and SDP-hierarchy integrality gaps, this condition on~$P$ has been significantly studied in the case of~$t = 3$.  For~$P$ supporting \pw uniformity, it is known~\cite{BGMT12,TW13} that the Sherali--Adams and Lov{\'a}sz--Schrijver$^+$  SDP hierarchies require degree~$\Omega(n)$ to refute random CSP$(P)$ instances (whp) when $m = O(n)$. This result was also recently proven for the stronger SOS proof system by Barak, Chan, and Kothari~\cite{BCK15}, except that the CSP$(P)$ instances are not quite uniformly random; they are ``slightly pruned'' random instances.   For any $t > 2$, the second and third authors recently essentially showed~\cite{OW14} that for the Sherali--Adams$^+$ SDP hierarchy, degree~$n^{\Omega(\eps)}$ is (whp) necessary to refute random CSP$(P)$ instances when $m \leq n^{t/2 - \eps}$.  As a caveat, again the instances are slightly pruned random instances, rather than being purely uniformly random.  (The instances in~\cite{OW14} are also in the ``Goldreich~\cite{Gol00} style''; i.e., there are no literals, but the ``right-hand sides'' are random.  However it is not hard to show the proofs in~\cite{OW14} go through in the standard random model of this paper.)  Future work~\cite{MWW15} is devoted to removing the pruning in these instances.  Although the Sherali--Adams$^+$ SDP hierarchy is certainly weaker than the SOS hierarchy, these works still constitute some evidence that our main theorem's requirement of $m \geq \wt{O}(n^{t/2})$ for non-\twus predicates may be essentially optimal.

Further evidence for the optimality of our results is provided by the work of Feldman, Perkins, and Vempala on statistical algorithms for random planted CSPs \cite{FPV15}.  They show that their lower bounds against statistical algorithms for solving random planted CSPs also imply lower bounds against statistical algorithms for refuting uniformly random CSPs.  Specifically, they prove that when $P$ supports a $(t-1)$-wise uniform distribution, any statistical algorithm using queries that take $\wt{O}(n^{t/2})$ possible values can only refute random instances of $\CSP(P)$ with at least $\wt{\Omega}(n^{t/2})$ constraints.  As an application of this result, they also show that any convex program refuting such an instance of $\CSP(P)$ must have dimension at least $\wt{\Omega}(n^{t/2})$.

\subsection{Certifying independence number and chromatic number of random hypergraphs}
Coja-Oghlan, Goerdt, and Lanka \cite{COGL07} also use their CSP refutation techniques to certify that random $3$- and $4$-uniform hypergraphs have small independence number and large chromatic number.  We extend these results to random $k$-uniform hypergraphs.
\begin{theorem} \label{thm:k-ind-set-intro}
For a random $k$-uniform hypergraph $H$, there is a polynomial time algorithm certifying that the independence number of $H$ is at most $\beta$ with high probability when $H$ has at least $\wt{O}\left(\frac{n^{5k/2}}{\beta^{2k}}\right)$ hyperedges.
\end{theorem}

\begin{theorem} \label{thm:k-chrom-num-intro}
For a random $k$-uniform hypergraph $H$, there is a polynomial time algorithm certifying that the chromatic number of $H$ is at least $\xi$ with high probability when $H$ has at least $\wt{O}\left(\xi^{2k} n^{k/2}\right)$ hyperedges.
\end{theorem}

The proofs of these theorems follow the outline of \cite{COGL07}.  We show Theorem~\ref{thm:k-ind-set-intro} using a slightly more general form of Theorem~\ref{thm:xor-intro}.  Theorem~\ref{thm:k-chrom-num-intro} follows almost directly from Theorem~\ref{thm:k-ind-set-intro} using the fact that every color class of a valid coloring is an independent set.  Details are given in Appendix~\ref{sec:hypergraphs}.


\section{Preliminaries and notation} \label{sec:prelims}
\subsection{Constraint satisfaction problems}           \label{sec:csp-prelims}
We review some basic definitions and facts related to constraint satisfaction problems (CSPs).  In this section we discuss only the Boolean domain, which we prefer to write as $\{-1,1\}$ rather than $\{0,1\}$.  The straightforward extensions to larger domains appear in Section~\ref{sec:larger-alphabets}.  We will need the following notation: For $x \in \R^n$ and $S \subseteq [n]$ we write $x_S \in \R^{|S|}$ for the restriction of $x$ to coordinates $S$; i.e., $(x_i)_{i \in S}$.  We also use $\circ$ to denote the entrywise product for vectors.
\begin{definition}
Given a predicate $P:\{-1,1\}^k \to \{0,1\}$, an instance $\calI$ of the $\maxkcsp(P)$ problem over variables $x_1, \dots, x_n$ is a multiset of \emph{$P$-constraints}.  Each $P$-constraint consists of a pair $(c, S)$, where $S \in [n]^k$ is the scope and $c \in \{-1,1\}^k$ is the negation pattern; this represents the constraint $P(c \circ x_S) = 1$.  We typically write $m = |\calI|$.  Let $\mathrm{Val}_{\calI}(x)$ be the fraction of constraints satisfied by assignment $x \in \{-1,1\}^n$, i.e., \ $\mathrm{Val}_{\calI}(x) = \frac{1}{m} \sum_{(c, S) \in \calI} P(c \circ x_S)$.  The objective is to find an assignment $x$ maximizing $\mathrm{Val}_{\calI}(x)$. The \emph{optimum} of $\calI$, denoted by $\val(\calI)$, is $\max_{x \in \kpm} \mathrm{Val}_{\calI}(x)$.  If $\val(\calI) = 1$, we say that $\calI$ is \emph{satisfiable}.  We also write $\expP$ for the quantity $\E_{z \sim \kpm}\left[P(z) \right]$; i.e., the fraction of assignments that satisfy~$P$.  For any instance $\calI$ in which each constraint involves $k$ different variables, we have $\Opt(\calI) \geq \expP$.\footnote{Technically, our definitions allow constraints with a variable appearing more than once, so $\Opt(\calI) \geq \expP$ doesn't always hold for us.  But since we only consider random~$\calI$, we'll in fact have $\Opt(\calI) \approx \expP$ whp over $\calI$ anyway.}
\end{definition}

We next define a standard random model for CSPs.  For $P:\{-1,1\}^k \to \{0,1\}$, let $\calF_P(n,p)$ be the distribution over CSP instances given by including each of the $2^k n^k$ possible constraints independently with probability $p$.  Note that we may include constraints on different permutations of the same set of variables, constraints on the same tuple of variables with different negations $c$, and constraints with the same variable occurring as more than one argument.  It is reasonable to include such constraints in the case that the predicate $P$ is not a symmetric function.  We use $\expm$ to denote the expected number of constraints, namely $2^k n^k p$.  As noted in Fact~\ref{fact:csp-facts} below, the number of constraints $m$ in a draw from $\calF_P(n,p)$ is very tightly concentrated around $\expm$, and we often blur the distinction between these parameters.  Appendix~\ref{sec:translation} explicitly describes a method for simulating an instance drawn from $\calF_P(n, p)$ when the number of constraints is fixed.

\paragraph{Quasirandomness.}  We now introduce an important notion for this paper: that of a CSP instance being \emph{quasirandom}.  Versions of this notion originate in the works of Goerdt and Lanka~\cite{GL03} (under the name ``discrepancy''), Khot~\cite{Kho06} (``quasi-randomness''), Austrin and H{\aa}stad~\cite{AH13} (``adaptive uselessness''), and Chan~\cite{Cha13} (``low correlation''), among other places.  To define it, we first need to define the induced distribution of an instance and an assignment.
\begin{definition}
\label{def:dist-2}
Given a CSP instance $\calI$ and and an assignment $x \in \{-1,1\}^n$, the \emph{induced distribution}, denoted $\calD_{\calI, x}$, is the probability distribution on $\{-1,1\}^k$ where the probability mass on $\alpha \in \{-1,1\}^k$ is given by
$
    \calD_{\calI, x}(\alpha) = \tfrac{1}{|\calI|} \cdot \#\{(c,S) \in \calI \mid c \circ x_S = \alpha\}.
$
In other words, it is the empirical distribution on inputs to~$P$ generated by the constraint scopes/negations on assignment~$x$.  Note that the predicate~$P$ itself is irrelevant to this notion. We will drop the subscript $\calI$ when it is clear from the context.  We define $D_{\calI, x} = 2^k \cdot \calD_{\calI, x}$ to be the density function associated with $\calD_{\calI, x}$.
\end{definition}
We can now define quasirandomness.
\begin{definition}
\label{def:quasi-2}
A CSP instance $\calI$ is \emph{$\eps$-quasirandom} if $\calD_{\calI,x}$ is $\eps$-close to the uniform distribution for all $x \in \{-1,1\}^n$; i.e., if
$
\dtv{\calD_{\calI, x}}{U^k} \leq \eps
$
for all $x \in \{-1,1\}^n$.
\end{definition}
Here we use the notation $U^k$~for the uniform distribution on $\{-1,1\}^k$ as well as the following:
\begin{definition}
    If $\calD$ and $\calD'$ are probability distributions on the same finite set~$A$ then $\dtv{\calD}{\calD'}$ denotes their \emph{total variation distance}, $\frac{1}{2} \sum_{\alpha \in A} \abs{\calD(\alpha) - \calD'(\alpha)}$.  If $\dtv{\calD}{\calD'} \leq \eps$ we say that $\calD$ and $\calD'$ are \emph{$\eps$-close}.  If $\dtv{\calD}{\calD'} \geq \eps$ we say they are \emph{$\eps$-far}.  (As neither inequality is strict, these notions are not quite opposites.)
\end{definition}

An immediate consequence of an instance being quasirandom is that its optimum is close to~$\expP$:
\begin{fact}
\label{fact:quasi-to-strong}
If $\calI$ is $\eps$-quasirandom, then $\val(\calI) \leq \expP + \eps$ (and in fact, $\abs{\val(\calI) - \expP} \leq \eps$).
\end{fact}

We conclude the discussion of CSPs by recording some facts that are proven easily with the Chernoff bound:
\begin{fact}
\label{fact:csp-facts}
    Let $\calI \sim \calF_{P}(n,p)$.  Then the following statements hold with high probability.
    \begin{enumerate}
        \item $m = \abs{\calI} \in \expm \cdot \left(1 \pm O\left(\sqrt{\frac{\log n}{\expm}}\right)\right)$.
        \item $\val(\calI) \leq \expP \cdot \left(1 + O\left(\sqrt{\frac{1}{\expP} \cdot \frac{n}{\expm}}\right)\right)$.
        \item $\calI$ is $O\left(\sqrt{2^k \cdot \frac{n}{\expm}}\right)$-quasirandom.
    \end{enumerate}
\end{fact}

\subsection{Algorithms and refutations on random CSPs}
\begin{definition}
    Let $P$ be a Boolean predicate.  We say that $\calA$ is \emph{$\delta$-refutation algorithm for random $\mathrm{CSP}(P)$ with $\expm$ constraints} if $\calA$ has the following properties.  First, on all instances $\calI$ the output of $\calA$ is either the statement ``$\Opt(\calI) \leq 1-\delta$'' or is ``fail''.  Second, $\calA$ is \emph{never} allowed to \emph{err}, where erring means outputting ``$\Opt(\calI) \leq 1-\delta$'' on an instance which actually has $\Opt(\calI) > 1-\delta$.  Finally, $\calA$ must satisfy
    \[
        \phantom{\quad \text{(as $n \to \infty$)},} \Pr_{\calI \sim \calF_P(n,p)}[\calA(\calI) = \text{``fail''}] < o(1) \quad \text{(as $n \to \infty$)},
    \]
    where $p$ is defined by $\expm = 2^k n^k p$.  Although $\calA$ is often a deterministic algorithm, we do allow it to be randomized, in which case the above probability is also over the ``internal random coins'' of~$\calA$.

    We refer to this notion as \emph{weak refutation}, or simply \emph{refutation}, when the certification statement is of the form ``$\Opt(\calI) < 1$'' (equivalently, when $\delta = 1/|\calI|$).  We refer to the notion as \emph{strong refutation} when the statement is of the form ``$\Opt(\calI) \leq \expP + o(1)$'' (equivalently, when $\delta = 1 - \expP + o(1)$).
\end{definition}
\begin{remark}
    In Section~\ref{sec:learning} we will encounter a ``two-sided error'' variant of this definition.  This is the slightly easier algorithmic task in which we relax the condition on erring: it is only required that for each instance $\calI$ with $\Opt(\calI) > 1-\delta$, it holds that $\Pr[\calA(\calI) = \text{``}\Opt(\calI) \leq 1-\delta\text{''}] \leq 1/4$, where the probability is just over the random coins of~$\calA$.
\end{remark}
\begin{remark}
    We will also use the analogous definition for certification of related properties; e.g., we will discuss \emph{$\eps$-quasirandomness certification algorithms} in which the statement ``$\Opt(\calI) \leq 1 - \delta$'' is replaced by the statement ``$\calI$ is $\eps$-quasirandom''.
\end{remark}

\subsection{$t$-wise uniformity}
An important notion for this paper is that of $t$-wise uniformity. Recall:
\begin{definition}
    Probability distribution $\calD$ on $\{-1,1\}^k$ is said to be \emph{$t$-wise uniform}, $1 \leq t \leq k$, if for all $S \subseteq [k]$ with $|S| = t$ the random variable $x_S$ is uniform on $\{-1,1\}^t$ when $x \sim \calD$.    (We remark that this condition is sometimes inaccurately called ``$t$-wise independence'' in the literature.)
\end{definition}
We will also consider the more general notion of $(\eps,t)$-wise uniformity.  This is typically defined using \emph{Fourier coefficients}:
\begin{definition}
    Probability distribution $\calD$ on $\{-1,1\}^k$ is said to be \emph{$(\eps,t)$-wise uniform} if $|\wh{D}(S)| \leq \eps$ for all $S \subseteq [k]$ with $0 < |S| \leq t$, where $D = 2^k \cdot \calD$ is the probability density associated with distribution~$\calD$.
\end{definition}
Here we are using standard notation from Fourier analysis of Boolean functions~\cite{OD14}.  In particular, for any $f \co \{-1,1\}^k \to \R$ we write $f(x) = \sum_{S \substeq [k]} \wh{f}(S) x^S$ for its expansion as a multilinear polynomial over~$\R$, with $x^S$ denoting $\prod_{i \in S} x_i$ (not to be confused with the projection $x_S \in \R^{|S|}$).
\begin{remark}  \label{rem:t-wise-triviality}
    It is a simple fact (and it follows from Lemma~\ref{lem:AGM} below) that $(0,t)$-wise uniformity is equivalent to $t$-wise uniformity.
\end{remark}

Also important for us is a related but distinct notion, that of being $\eps$-close to a $t$-wise uniform distribution.   It's easy to show that if $\calD$ is $\eps$-close to a $t$-wise uniform distribution then $\calD$ is $(2\eps, t)$-wise \todo{did I get that $2\eps$ correct?}uniform.  In the other direction, we have the following (see also~\cite{AAK+07} for some quantitative improvement):
\begin{lemma}                                     \label{lem:AGM}
    \textup{(Alon--Goldreich-Mansour~\cite[Theorem~2.1]{AGM03}).}  If $\calD$ is an $(\eps,t)$-wise uniform distribution on $\{-1,1\}^k$, then there exists a $t$-wise uniform distribution $\calD'$ on $\{-1,1\}^k$ with
    \[
        \dtv{\calD}{\calD'} \leq \Bigl(\sum_{i=1}^t \tbinom{k}{t}\Bigr) \cdot \eps \leq k^t \cdot \eps.
    \]
    In particular if $t = k$ we have the  bound $2^k \cdot \eps$ (and this can also be improved~\cite{Gol11} to $2^{k/2-1} \cdot    \eps$).
\end{lemma}

Finally, we make a crucial definition:
\begin{definition}
    A predicate $P : \{-1,1\}^k \to \{0,1\}$ is said to be \emph{\twus} if there is a $t$-wise uniform distribution $\calD$ whose support is contained in $P^{-1}(1)$.  We say $P$ is \emph{$\delta$-far from \twus} if every $t$-wise uniform distribution $\calD$ is $\delta$-far from being supported on~$P$; i.e., has probability mass at least~$\delta$ on $P^{-1}(0)$.
\end{definition}

\subsection{A dual characterization of limited uniformity}
\label{sec:dual-polys}
It is known that the condition of $P$ supporting a $t$-wise uniform distribution is equivalent to the feasibility of a certain linear program; hence one can show that $P$ is \emph{not} \twus by exhibiting a certain dual object, namely a polynomial.  This appears, e.g., in work of Austrin and H{\aa}stad~\cite[Theorem~3.1]{AH09}. Herein we will extend this fact by giving a dual characterization of being far from \twus.

\begin{definition}
\label{def:poly-sep}
Let $0 < \delta < 1$.  For a multilinear polynomial $Q: \{-1, 1\}^k \rightarrow \R$, we say that $Q$ \emph{$\delta$-separates} $P:\{-1, 1\}^k \rightarrow \{0, 1\}$ if the following conditions hold:
\begin{itemize}
\item $Q(z) \geq \delta-1 \quad \forall z\in \{-1, 1\}^k$;
\item $Q(z) \geq \delta \quad \forall z \in P^{-1}(1)$;
\item $\wh{Q}(\emptyset) = 0$, i.e., $Q$ has no constant coefficient.
\end{itemize}
\end{definition}

We now provide the quantitative version of the aforementioned~\cite[Theorem~3.1]{AH09}:
\begin{lemma}
\label{lem:poly-iff-no-dist}
Let \Pdef\ and let $0 < \delta < 1$.  Then $P$ is $\delta$-far from \twus if and only if there is a $\delta$-separating polynomial for~$P$ of degree at most~$t$.
\end{lemma}
\begin{proof}
The proof is an application of linear programming duality.  Consider the following LP, which has variables $\calD(z)$ for each $z \in \{-1,1\}^k$.
\begin{center}
\fbox{
\begin{minipage}{0.8\textwidth}
\begin{alignat}{3}
    \textrm{minimize}            &   & \sum_{\mc{z \in \{-1, 1\}^k}}(1 - &P(z))\calD(z)    \label{eq:primal-lp}     &  &      \\
    \text{s.t.}\qquad & &\sum_{\mc{z\in\kpm}}\calD(z)z^S  = 2^k \cdot \wh{\calD}(S)           & =            0                  && \forall S\subseteq [k]  \quad0 < |S| \leq t  \label{eq:uniform}\\
                      &   &       \sum_{{z \in \kpm}}\calD(z)             & =            1                  &   &      \label{eq:sum1} \\
                      &   & \calD(z)         & \geq            0    &&\forall z \in \kpm \nonumber
 \end{alignat}
\end{minipage}
}
\end{center}
Constraint~\eqref{eq:sum1} and the nonnegativity constraint ensure that $\calD$ is a probability distribution on~$\kpm$.  Constraint~\eqref{eq:uniform} expresses that $\calD$ is $t$-wise uniform (see Remark~\ref{rem:t-wise-triviality}).  The objective~\eqref{eq:primal-lp} is minimizing the probability mass that $\calD$ puts on assignments in~$P^{-1}(0)$.  Thus the optimal value of the LP is equal to the smallest~$\gamma$ such that $P$ is $\gamma$-close to \twus; equivalently, the largest~$\gamma$ such that $P$ is $\gamma$-far from \twus.

The following is the dual of the above LP.  It has a variable $c(S)$ for each $0 < |S| \leq t$ as well as a variable $\xi$ corresponding to constraint~\eqref{eq:sum1}.
\begin{center}\fbox{
\begin{minipage}{0.7\textwidth}
\begin{alignat}{3}
    \textrm{maximize}            &  & \xi &  \label{eq:dual-lp}     &  &      \\
    \text{s.t.}\qquad                 &   &       \sum_{\substack{S \subseteq [k]\\0 < |S| \leq t}}c(S) z^S  &\leq 1-P(z) -\xi  &\qquad  &\forall z \in \kpm.      \label{eq:strict-ineq}
     \end{alignat}
\end{minipage}}
\end{center}
Observe that a feasible solution $(\{c(S)\}_S, \xi)$ is precisely equivalent to a multilinear polynomial $Q$ of degree at most~$t$, namely $Q(z) = -\sum_S c(S) z^S$, that $\xi$-separates~$P$.

Thus $P$ is $\delta$-far from \twus if and only if the LP's objective~\eqref{eq:primal-lp} is at least~$\delta$, if and only if the dual's objective~\eqref{eq:dual-lp} is at least~$\delta$, if and only if there is a $\delta$-separating polynomial for~$P$ of degree at most~$t$.
\end{proof}

From this proof we can also derive that if $P$ fails to be \twus then it must in fact be $\Omega_k(1)$-far from being \twus:
\begin{corollary}                                       \label{cor:lp-granularity}
    Suppose $P \co \{-1,1\}^k \to \{0,1\}$ is not \twus. Then it is in fact $\delta$-far from \twus for $\delta = 2^{-\wt{O}(k^t)}$ (or $\delta = 2^{-\wt{O}(2^k)}$ when $t = k$).
\end{corollary}
\begin{proof}
    Let $K = 1+\sum_{i = 1}^t \binom{k}{t}$ be the number of variables in the  dual LP from Lemma~\ref{lem:poly-iff-no-dist}, so $K \leq k^t+1$ in general, with $K \leq 2^k$ when $t = k$.  By assumption, the objective~\eqref{eq:dual-lp} of the dual LP's optimal solution is  strictly positive. This optimum occurs at a vertex, which is the solution of a linear system given by a $K \times K$ matrix of $\pm 1$ entries and a ``right-hand side'' vector with $0,1$ entries.  By Cramer's rule, the solution's entries are ratios of determinants of integer matrices with entries in $\{-1,0,1\}$.  Thus any strictly positive entry is at least $1/N$, where $N$ is the maximum possible such determinant.  By Hadamard's inequality, $N = K^{K/2}$ and the claimed result follows.
\end{proof}


\section{Quasirandomness and its implications for refutation}
\label{sec:quasi}

\subsection{Strong refutation of $k$-XOR}
\label{sec:refute-xor}
In this section, we state our result on strong refutation of random $k$-XOR instances with $m = \tilde{O}\left(n^{k/2}\right)$ constraints.  (Recall that essentially this result was very recently obtained by Barak and Moitra~\cite{BM15}.)  We actually have a slightly more general result, allowing variables and coefficients to take values in $[-1,1]$ and not just in $\{-1,1\}$.  We will use this additional freedom to prove refutation results for CSPs over larger alphabets in Appendix~\ref{sec:larger-alphabets} and refutation results for independence number and chromatic number of random hypergraphs in Appendix~\ref{sec:hypergraphs}.
\begin{theorem}
\label{thm:main}
For $k \geq 2$ and $p \geq n^{-k/2}$, let $\{w(T)\}_{T \in [n]^k}$ be independent random variables such that for each $T \in [n]^k$,
\begin{align}
\E[w(T)] &= 0 \label{eq:w-mean-0} \\
\Pr[w(T) \ne 0] &\leq p \label{eq:w-mostly-0} \\
\abs{w(T)} &\leq 1. \label{eq:w-bounded}
\end{align}
Then there is an efficient algorithm certifying that
\begin{equation*}
\sum_{T \in [n]^k} w(T) x^T \leq  2^{O(k)} \sqrt{p} n^{3k/4} \log^{3/2} n.
\end{equation*}
for all $x \in \R^n$ such that $\norm{x}_{\infty} \leq 1$ with high probability.
\end{theorem}
In this form, the theorem is not really about CSP refutation at all.  It says that the value of a polynomial with random coefficients is close to its expectation when its inputs are bounded.

We give the proof in Appendix~\ref{sec:xor}.  It follows techniques from~\cite{COGL07} fairly closely and is essentially the same as the proof of~\cite{BM15}.  We will use this theorem to prove our results in subsequent sections.

We obtain strong refutation of $k$-XOR as a simple corollary.
\begin{corollary}
\label{cor:strong-xor-ref}
For $k \geq 2$, let $\calI \sim \calF_{\textup{$k$-XOR}}(n,p)$.  Then, with high probability, there is a degree-$2k$ SOS proof that $\Opt(\calI) \leq \frac{1}{2} + \gamma$ when $\expm \geq \frac{2^{O(k)}n^{k/2} \log^3 n}{\gamma^2}$.
\end{corollary}

\begin{proof}
We can write the $k$-XOR predicate as
\[
\text{$k$-XOR}(z) = \frac{1-\prod_{i=1}^k z_i}{2},
\]
so for a $k$-XOR instance $\calI \sim \calF_\textup{$k$-XOR}(n,p)$,
\[
\textrm{Val}_\calI(x) = \frac{1}{2} - \frac{1}{2m} \sum_{T \in [n]^k} \sum_{c \in \{\pm 1\}^k} \indic{\{(T,c) \in \calI\}} x^T \prod_{i \in [k]} c_i \;=\; \frac{1}{2} + \frac{2^{k-1}}m \sum_{T \in [n]^k} w(T) x^T,
\]
where $w(T) = -2^{-k} \sum_{c \in \{\pm 1\}^k} \indic{\{(T,c) \in \calI\}} \prod_{i \in [k]} c_i$.  The $w(T)$'s are random variables depending on the choice of $\calI$; observe that $\E[w(T)] = 0$, $\Pr[w(T) \ne 0] \leq 2^kp$, and $\abs{w(T)} \leq 1$ for all $T \in [n]^k$.   By Theorem~\ref{thm:main}, there is an algorithm certifying that
\[
\val(\calI) \leq \frac{1}{2} + \frac{2^{O(k)} \sqrt{p} n^{3k/4} \log^{3/2} n}{m}.
\]
with high probability when $p \geq n^{-k/2}$.  Since $m = (1+o(1))\expm$ with high probability, choosing $\expm \geq \frac{2^{O(k)}n^{k/2} \log^3 n}{\gamma^2}$ gives the desired result.
\end{proof}

As an example, we can choose $\gamma = \frac{1}{\log n}$ and certify that $\val(\calI) \leq \frac{1}{2} + o(1)$ when $\expm = \wt{O}_k(n^{k/2})$.

\subsection{Quasirandomness and strong refutation of any $k$-CSP}
Next, we will use the algorithm of Theorem~\ref{thm:main} to certify that an instance of CSP$(P)$ is quasirandom.  This will immediately give us a strong refutation algorithm.

In order to certify quasirandomness, Lemma~\ref{lem:AGM} implies that it suffices to certify each Fourier coefficient of $D_{\calI,x}$ has small magnitude.
\begin{lemma}
\label{lem:fourier-refute}
Let $\emptyset \ne S \subseteq [k]$ with $|S| = s$.  There is an algorithm that, with high probability, certifies that
\[
\abs{\widehat{D_{\calI,x}}(S)} \leq \frac{2^{O(s)} \max\{n^{s/4}, \sqrt{n}\} \log^{5/2} n}{\expm^{1/2}}
\]
for all $x \in \{-1,1\}^n$, assuming also that $\expm \geq \max\{n^{s/2}, n\}$.
\end{lemma}
To prove this lemma, we need another lemma certifying that polynomials whose coefficients are sums of $0$-mean random variables have small value.
\begin{lemma}
\label{lem:smaller-monomials}
Let $S \subseteq [k]$ with $|S| = s > 0$.  Let $\tau \in \N$ and let $\{w_U(i)\}_{U \in [n]^s, i \in [\tau]}$ be independent random variables satisfying conditions \eqref{eq:w-mean-0}, \eqref{eq:w-mostly-0}, and \eqref{eq:w-bounded} for some $p \geq \frac{1}{\tau n^{s/2}}$.  Then there is an algorithm that certifies with high probability that
\[
\sum_{U \in [n]^s} x^U \sum_{j = 1}^{\tau} w_U(j) \leq
\begin{cases}
2^{O(s)} \sqrt{\tau p} \cdot n^{3s/4} \log^{5/2} n \quad & \text{if $s \geq 2$} \\
4 \max\{\sqrt{\tau p}, 1\} \cdot n \log n \quad & \text{if $s = 1$.}
\end{cases}
\]
for all $x \in \R^n$ such that $\norm{x}_{\infty} \leq 1$.
\end{lemma}
The proof is straightforward and we defer it to Section~\ref{sec:smaller-monomials}.  

\begin{proof}[Proof of Lemma~\ref{lem:fourier-refute}]
Without loss of generality, assume $1 \in S$.  Applying definitions, we see that
\begin{equation}
\label{eq:induced-fourier}
\widehat{D_{\calI,x}}(S) = \E_{\mc{z \sim \calD_{\calI,x}}}\left[z^S\right] = \frac{1}{m} \sum_{U \in [n]^s} \sum_{\substack{T \in [n]^k \\ T_S = U}} \sum_{c \in \{\pm 1\}^{k}} \indic{\{(T,c) \in \calI\}} c^S x^U = \frac{1}{m} \sum_{U \in [n]^s} x^U \sum_{\substack{T \in [n]^k \\ T_S = U}} \sum_{c' \in \{\pm 1\}^{k-1}} w_S(T,c').
\end{equation}
where we define $w_S(T,c') = \indic{\{(T,(1,c')) \in \calI\}} (c')^{S \setminus \{1\}} - \indic{\{(T,(-1,c')) \in \calI\}} (c')^{S \setminus \{1\}}$ and recall that $T_S$ is the projection of $T$ onto the the coordinates in $S$.
It is clear that $\E[w_S(T,c')] = 0$, $\Pr[w_S(T,c') \ne 0] \leq p$, and $|w_S(T,c')| \leq 1$.  There are $\tau = 2^{k-1} n^{k-s}$ terms in each sum of $w_S(T,c')$'s and we can apply Lemma~\ref{lem:smaller-monomials}.  When $s = 2$, we plug in these values and see that we can certify that $\widehat{D_{\calI,x}}(S) \leq \frac{2^{O(s)} n^{s/4} \log^{5/2} n}{\expm^{1/2}}$.  When $s = 1$, $\expm \geq n$ implies that $\tau p \geq \frac{1}{2}$ and we can certify that $\widehat{D_{\calI,x}}(S) \leq \frac{2^{O(s)} \sqrt{n} \log n}{\expm^{1/2}}$.  The lower bound can be proved in exactly the same way by considering the random variables $-w_S(T,c')$.
\end{proof}

The existence of an algorithm for certifying quasirandomness follows from Lemmas~\ref{lem:AGM}~and~\ref{lem:fourier-refute}.
\begin{theorem}
\label{thm:quasi-2}
There is an efficient algorithm that certifies that an instance $\calI \sim \calF_P(n,p)$ of $\CSP(P)$ is $\gamma$-quasirandom with high probability when $\expm \geq \frac{2^{O(k)} n^{k/2} \log^{5} n}{\gamma^2}$.
\end{theorem}
Since $\gamma$-quasirandomess implies that $\val(\calI) \leq \expP + \gamma$, this algorithm also strongly refutes $\CSP(P)$.
\begin{theorem}
\label{thm:strong-ref-2}
There is an efficient algorithm that, given an instance $\calI \sim \calF_P(n,p)$ of $\CSP(P)$, certifies that $\val(\calI) \leq \expP + \gamma$ with high probability when $\expm \geq \frac{2^{O(k)} n^{k/2} \log^{5} n}{\gamma^2}$.
\end{theorem}

\subsection{$(\eps,t)$-quasirandomness and $\Omega(1)$-refutation of non-$t$-wise-supporting CSPs}
In the case that a predicate is not \twus, a weaker notion of quasirandomness suffices to obtain $\Omega(1)$-refutation.
\begin{definition}
An instance $\calI$ of $\csp(P)$ is $(\eps,t)$-quasirandom if $\calD_{\calI,x}$ is $(\eps,t)$-wise uniform for every $x \in \{-1,1\}^n$.
\end{definition}
Fact~\ref{fact:csp-facts} shows that random instances with $\wt{O}(n)$ constraints are $(o(1),t)$-quasirandom for all $t \leq k$ with high probability.  Lemma~\ref{lem:fourier-refute} directly implies that we can certify $(\eps,t)$-quasirandomness when $m \geq \wt{O}(n^{t/2})$.
\begin{theorem}
\label{thm:t-quasi-2}
There is an efficient algorithm that certifies that an instance $\calI \sim \calF_P(n,p)$ of \textup{CSP}$(P)$ is $(\gamma,t)$-quasirandom with high probability when $\expm \geq \frac{2^{O(t)} n^{t/2} \log^{5} n}{\gamma^2}$ and $t \geq 2$.
\end{theorem}

We now reach the main result of this section, which states that if a predicate is $\delta$-far from \twus, then we can almost $\delta$-refute instances of CSP$(P)$.
\begin{theorem}
\label{thm:ntwus-refute-2}
Let $P$ be $\delta$-far from \twus.  Then there is an efficient algorithm that, given an instance $\calI \sim \calF_P(n,p)$ of \textup{CSP}$(P)$,  certifies that $\val(\calI) \leq 1 - \delta + \gamma$ with high probability when $\expm \geq \frac{k^{O(t)} n^{t/2} \log^{5} n}{\gamma^2}$ and $t \geq 2$.
\end{theorem}

We give two proofs of this theorem.  In Proof~1, the theorem follows directly from certification of $(\gamma,t)$-quasirandomness and Lemma~\ref{lem:AGM}.
\begin{proof}[Proof 1]
Run the algorithm of Theorem~\ref{thm:t-quasi-2} to certify that $\calI$ is $(\gamma/k^t,t)$-quasirandom with high probability.  By definition, we have certified that $\calD_{\calI,x}$ is $(\gamma/k^t, t)$-wise uniform for all $x \in \{-1,1\}^n$,.  Lemma~\ref{lem:AGM} then implies that for all $x$ there exists a $t$-wise uniform distribution $\calD'_x$ such that $\dtv{\calD_{\calI ,x}}{\calD'_x} \leq \gamma$.  Now define $\calD_{\mathrm{sat}}$ to be an arbitrary distribution over satisfying assignments to $P$.  Since $P$ is $\delta$-far from being \twus, we know that $\dtv{\calD'}{\calD_{\mathrm{sat}}} \geq \delta$ for any $t$-wise uniform distribution $\calD'$.  The triangle inequality then implies that $\dtv{\calD_{\calI ,x}}{\calD_{\mathrm{sat}}} \geq \delta - \gamma$ for all $x \in \{-1,1\}^n$ and the theorem follows.
\end{proof}

Proof~2 gives a slightly weaker version of Theorem~\ref{thm:ntwus-refute-2}, requiring the stronger assumption that $\expm \geq \frac{2^{O(k)} n^{t/2} \log^{5} n}{\gamma^2}$.  It is based on the dual polynomial characterization of being $\delta$-far from $t$-wise supporting.  While perhaps less intuitive than Proof~1, Proof~2 is more direct.  It only uses the XOR refutation algorithm and bypasses \cite{AGM03}'s connection between $(\eps,t)$-wise uniformity and $\eps$-closeness to a $t$-wise uniform distribution.  We were able to convert Proof~2 into an SOS proof (see Section~\ref{sec:sos-ntwus}), but we did not see how to translate Proof~1 into an SOS version. Proof~2 requires Plancherel's Theorem, a fundamental result in Fourier analysis.
\begin{theorem}[Plancherel's Theorem] \label{thm:plancherel}
For any $f,g:\{-1,1\}^k \to \R$,
\[
\E_{\bz \in U^k}[f(\bz)g(\bz)] = \sum_{S \subseteq [k]} \wh{f}(S) \wh{g}(S).
\]
\end{theorem}

\begin{proof}[Proof 2]
Since $P$ is $\delta$-far from $t$-wise supporting, there exists a degree-$t$ polynomial $Q$ that $\delta$-separates $P$.  The definition of $\delta$-separating implies that $P(z) - (1-\delta) \leq Q(z)$ for all $z \in \{-1,1\}^k$.  Summing over all constraints, we get that for all $x \in \{-1,1\}^n$,
\[
\sum_{T \in [n]^k} \sum_{c \in \{\pm 1\}^k} \indic{\{(T,c) \in \calI\}} P(x_T \circ c) - m(1-\delta) \leq \sum_{T \in [n]^k} \sum_{c \in \{\pm 1\}^k} \indic{\{(T,c) \in \calI\}} Q(x_T \circ c),
\]
or, equivalently, $\mathrm{Val}_{\calI}(x) - (1-\delta) \leq \E_{\bz \in \calD_{\calI,x}}[Q(\bz)]$.

It then remains to certify that $\E_{\bz \in \calD_{\calI,x}}[Q(\bz)] \leq \gamma$.  Observe that
\[
\E_{\bz \in \calD_{\calI,x}}[Q(\bz)] = \E_{\bz \in U^k}[D_{\calI,x}(\bz)Q(\bz)] = \sum_{\emptyset \ne S \subseteq [k]} \wh{D_{\calI,x}}(S) \wh{Q}(S),
\]
where the second equality follows from Plancherel's Theorem.  Since $Q \geq -1$ and $\E[Q] = 0$, $Q \leq 2^k$ and hence $|\wh{Q}(S)| \leq 2^k$ for all $S$.  To finish the proof, we apply Lemma~\ref{lem:fourier-refute} to certify that $\left|\widehat{D_{\calI,x}}(S)\right| \leq \frac{\gamma}{2^{2k}}$ for all $S$.
\end{proof}

With Corollary~\ref{cor:lp-granularity}, Theorem~\ref{thm:ntwus-refute-2} implies that we can $\Omega_k(1)$-refute instances of $\CSP(P)$ with $\wt{O}_k(n^{t/2})$ constraints when $P$ is not $t$-wise supporting.
\begin{corollary}
\label{cor:no-t-wise-ref-2}
Let $P$ be a predicate that does not support any $t$-wise uniform distribution.  Then there is an efficient algorithm that, given an instance $\calI \sim \calF_P(n,p)$ of $\CSP(P)$, certifies that $\val(\calI) \leq 1 - 2^{-\wt{O}(k^t)}$ with high probability when $\expm \geq 2^{\wt{O}(k^t)} n^{t/2} \log^{5} n$ and $t \geq 2$.
\end{corollary}

\subsection{Proof of Lemma~\ref{lem:smaller-monomials}}
\label{sec:smaller-monomials}
Recall the statement of the lemma.
\begin{customlem}{\ref{lem:smaller-monomials}}
Let $S \subseteq [k]$ with $|S| = s > 0$.  Let $\tau \in \N$ and let $\{w_U(i)\}_{U \in [n]^s, i \in [\tau]}$ be independent random variables satisfying conditions \eqref{eq:w-mean-0}, \eqref{eq:w-mostly-0}, and \eqref{eq:w-bounded} for some $p \geq \frac{1}{\tau n^{s/2}}$.  Then there is an algorithm that certifies with high probability that
\[
\sum_{U \in [n]^s} x^U \sum_{j = 1}^{\tau} w_U(j) \leq
\begin{cases}
2^{O(s)} \sqrt{\tau p} \cdot n^{3s/4} \log^{5/2} n \quad & \text{if $s \geq 2$} \\
4 \max\{\sqrt{\tau p}, 1\} \cdot n \log n \quad & \text{if $s = 1$.}
\end{cases}
\]
for all $x \in \R^n$ such that $\norm{x}_{\infty} \leq 1$.
\end{customlem}
The proof uses Bernstein's Inequality.
\begin{theorem}[Bernstein's Inequality]
\label{thm:bernstein}
Let $X_1, \ldots, X_M$ be independent $0$-mean random variables such that $\abs{X_i} \leq B$.  Then, for $a > 0$,
\[
\Pr\left[\sum_{i = 1}^M X_i > a\right] \leq \exp\left(\frac{-\frac{1}{2}a^2}{\sum_{i = 1}^M \E[X_i^2] + \frac{1}{3}Ba}\right).
\]
\end{theorem}

\begin{proof}[Proof of Lemma~\ref{lem:smaller-monomials}]
First, we define
\[
v_U = \sum_{j = 1}^{\tau} w_U(i).
\]
Observe that the $v_U$'s are independent and that each one is the sum of $\tau$ mean-$0$, i.i.d. random variables with magnitude at most~$1$.  Noting that $\sum_{i = 1}^{\tau} \E[w_U(i)^2] \leq \tau p$, we can use Bernstein's Inequality to show that the $|v_U|$'s are not too big with high probability.  If $s \geq 2$, Theorem~\ref{thm:main} then implies that the desired algorithm exists.  If $s = 1$, we are simply bounding a linear function over $\pm 1$ variables.  We consider two cases:  Small $p$ and large $p$.
\paragraph{Case 1: $p \leq \frac{1}{4\tau}$.}  Choosing $a = 2s\log n$ in Bernstein's Inequality, we see that $\Pr[|v_U| \geq 2s\log n] \leq n^{-2s}$. A union bound over all $U$ then implies that $\Pr[\text{any $|v_U| > 2s \log n$}] \leq n^{-s}$.  If $s \geq 2$, we observe that $\Pr[v_U \neq 0] \leq \tau p$, scale the $v_U$'s down by $2s \log n$, and apply Theorem~\ref{thm:main} to get the stated result.  If $s = 1$, we obtain the second bound by observing that
\begin{equation}
\label{eq:1-case}
\sum_{i \in [n]} v_i x_i \leq \sum_{i \in [n]} \abs{v_i} \leq 2n \log n.
\end{equation}
\paragraph{Case 2: $p > \frac{1}{4 \tau}$.} We set $a = 4 s \sqrt{\tau p} \log n$ and get that $\Pr[\text{any $\abs{v_U} > 4 s \sqrt{\tau p} \log n$}] \leq n^{-s}$ as above.  If $s \geq 2$, we can then divide the $v_U$'s by $4s \sqrt{\tau p} \log n$ and apply Theorem~\ref{thm:main}.  If $s = 1$, we get a bound of $4 \sqrt{\tau p} \cdot n \log n$ in the same way as \eqref{eq:1-case}.
\end{proof}


\section{Hardness of learning implications}\label{sec:learning}
Recent work by Daniely et al.~\cite{daniely-hardness-of-learning} reduces the problem of refuting specific instances of $\csp(P)$ to the problem of improperly learning certain hypothesis classes in the Probably Approximately Correct (PAC) model~\cite{Valiant:1984:TL:1968.1972}.
In this model, the learner is given $m$ labeled training examples $(x_1, \ell(x_1)), \ldots, (x_m, \ell(x_m))$, where each $x_i \in \{-1, 1\}^n$, each $\ell(x_i) \in \{0, 1\}$, and the examples are drawn from some unknown distribution $\calD$ on $\{-1, 1\}^n \times \{0, 1\}$.  For some hypothesis class $\calH \subseteq \{0, 1\}^{\{-1, 1 \}^n}$ we say that $\calD$ can be \emph{realized} by $\calH$ if there exists some $h \in \calH$ such that $\Pr_{(x, \ell(x))\sim\calD}[h(x) \neq \ell(x)] = 0$.
In improper PAC learning, on an input of $m$ training examples drawn from $\calD$ such that $\calD$ can be realized by some $h \in \calH$,
and an error parameter $\epsilon$,
the algorithm 
outputs some hypothesis function $f_h:\{-1, 1\}^n \rightarrow \{0, 1\}$ (not necessarily in $\calH$) such that $\Pr_{(x, \ell(x))\sim\calD}[f_h(x) \neq \ell(x)] \leq \eps$.
In improper \emph{agnostic} PAC learning, the assumption that $\calD$ can be realized by some $h \in \calH$ is removed and the algorithm must output a hypothesis that performs almost as well as the best hypothesis in $\calH$.
More formally, the hypothesis $f_h$ must satisfy the following: $\Pr_{(x, \ell(x))\sim\calD}[f_h(x) \neq \ell(x)] \leq \min_{h \in \calH}\Pr_{(x, \ell(x))\sim\calD}[h(x) \neq \ell(x)] + \eps$.
In improper approximate agnostic PAC learning, the learner is also given an approximation factor $a \geq 1$ and must output a hypothesis $f_h$ such that $\Pr_{(x, \ell(x))\sim\calD}[f_h(x) \neq \ell(x)] \leq a \cdot \min_{h \in \calH}\Pr_{(x, \ell(x))\sim\calD}[h(x) \neq \ell(x)] + \eps$.

Daniely et al.~reduce the problem of distinguishing between random instances of $\csp(P)$ and instances with value at least $\alpha$ as a PAC learning problem by transforming each constraint into a labeled example.
To show hardness of improperly learning a certain hypothesis class in the PAC model, they define a predicate $P$ that is specific to the hypothesis class and assume hardness of distinguishing between random instances of $\csp(P)$ and instances with $n^d$ constraints and value at least $\alpha$ for all $d > 0$.
They then demonstrate that the sample can be realized (or approximately realized) by some function in the hypothesis class if the CSP instance is satisfiable (or has value at least $\alpha$).
They also show that if the given CSP instance is random, the set of examples will have error at least $\tfrac14$ (in the agnostic case $\tfrac15)$ for all $h$ in the hypothesis class with high probability.
Using this approach, they obtain hardness results for the following problems:  improperly learning DNF formulas, improperly learning intersections of 4 halfspaces, and improperly approximately agnostically learning halfspaces for any approximation factor.

\subsection{Hardness assumptions}
The hardness assumptions made in~\cite{daniely-hardness-of-learning} are the same as those presented in Section~\ref{ssec:ltprelims}, except for a few minor differences.
First, their model fixes the number of constraints rather than the probability with which each constraint is included in the instance.  It is well-known that results in one model easily translate to the other.  We include a proof in Appendix~\ref{sec:translation} for completeness.  Additionally, SRCSP Assumptions 1 and 2 purport hardness of distinguishing random instances of $\csp(P)$ from satisfiable instances, even when the algorithm is allowed to err with probability $\tfrac14$ over its internal coins.  The algorithms presented in the preceding sections never err on satisfiable instances; further,  they only fail to certify random instances with probability $o(1)$.  As a result,  our refutation algorithms also falsify weaker versions of both SRCSP Assumptions, wherein the allowed probability of error is both lower and one-sided.
For each predicate presented in~\cite{daniely-hardness-of-learning}, we falsify the appropriate SRCSP assumption using the following approach.  For each predicate $P$ and corresponding $\delta > 0$ , we define a degree-$t$ polynomial that $\delta$-separates $P$.  Using the refutation techniques presented in the preceding sections, we deduce that $\wt O(n^{t/2})$ constraints are sufficient to distinguish random instances of $\csp(P)$ from those that are satisfiable (or have value at least $\alpha$).
In order to simplify the presentation, we begin with simpler versions of the polynomials and then scale them to attain the appropriate values of $\delta$.
The following lemma will be of use for this scaling.

\begin{lemma}\label{lem:scaleQ}
For predicate $P:\{-1, 1\}^k \rightarrow \{0, 1\}$, let $\Qdef$ be an unbiased multilinear polynomial of degree $t$ such that there exist $\theta_1 > 0, \theta_0<0$ not dependent on $z$ for which the following holds: $Q(z) \geq \theta_1$ for all $z \in P^{-1}(1)$ and $Q(z) \geq \theta_0$ for all $z \in \kpm$.
Then there exists a degree-$t$ polynomial $\calQ: \kpm \rightarrow \R$ that $\frac{\theta_1}{\theta_1 - \theta_0}$-separates $P$.
\end{lemma}
\begin{proof}
Define $\calQ(z) = \frac{Q(z)}{\theta_1 - \theta_0}$.  Clearly $\calQ$ is also unbiased and has degree $t$.
Then for all $z \in P_1$,  $\frac{Q(z)}{\theta_1 - \theta_0} \geq \frac{\theta_1}{\theta_1  - \theta_0}$.
Similarly, for all $z$, $\frac{Q(z)}{\theta_1 - \theta_0} \geq \frac{\theta_0}{\theta_1 - \theta_0} = -\frac{\theta_1 - \theta_0}{\theta_1 - \theta_0} + \frac{\theta_1}{\theta_1 - \theta_0} = -1 + \frac{\theta_1}{\theta_1 - \theta_0}$.
\end{proof}

We now demonstrate that the above can be applied to the predicates suggested in~\cite{daniely-hardness-of-learning} by defining separating polynomials and applying Theorem~\ref{thm:ntwus-refute-2}

\subsection{Huang's predicate and hardness of learning DNF formulas}
In order to obtain hardness of improperly learning DNF formulas with $\omega(1)$ terms, Daniely et~al.\ use the following predicate, introduced by Huang~\cite{Huang:2013:ARS:2488608.2488666}.  Huang showed that it is hereditarily approximation resistant; Daniely et~al.\ also observed that its $0$-variability is $\Omega(k^{1/3})$~\cite{daniely-hardness-of-learning}.

\begin{definition}
Let $k = \kappa + {\kappa \choose 3}$ for some integer $\kappa \geq 3$.
For $z \in \kpm$, index $z$ as follows.  Label the first $\kappa$ bits of $z$ as $z_1, \ldots, z_{\kappa}$.
The remaining ${\kappa \choose 3}$ bits are indexed by unordered triples of integers between $1$ and $\kappa$. Each $T \subseteq [\kappa]$ with $|T| =3$ is associated with a distinct bit of the remaining ${\kappa \choose 3}$ bits, which is indexed by $z_T$.
We say that $z$ \emph{strongly satisfies} the Huang predicate iff for every $T = \{z_i, z_j, z_\ell\}$ such that  $z_i, z_j, z_\ell$ are distinct elements of $[\kappa]$, $z_iz_jz_\ell = z_{\{i,j,\ell\}}$.
Additionally, we say that $z$ \emph{satisfies} the Huang predicate iff there exists some $z' \in \kpm$ such that $z$ has Hamming distance at most $\kappa$ from $z'$ and $z'$ strongly satisfies the Huang predicate.
Define $H_{\kappa}: \kpm \rightarrow\{0, 1\}$ as follows:  $H_{\kappa}(z) = 1$ if $z$ satisfies the Huang predicate and $H_{\kappa}(z) = 0$ otherwise.
\end{definition}

Daniely et al.~reduce the problem of distinguishing between random instances of $\csp(H_{\kappa})$ with $2n^d$ constraints and satisfiable instances to the problem of improperly PAC learning the class of DNF formulas with $\omega(1)$ terms on a sample of $O(n^d)$ training examples with error $\eps = 1/5$ with probability at least $\tfrac34$. Here we show that there exists a polynomial time algorithm that refutes random instances of $\csp(H_\kappa)$ by demonstrating that $H_k$ does not support a 4-wise uniform distribution and applying Theorem~\ref{thm:ntwus-refute-2}.

\begin{theorem}
Assume $\kappa \geq 9$.
There exists a degree-$4$ polynomial $\calQ: \kpm \rightarrow \R$ that 
$\tfrac18$-separates $H_\kappa$.
Consequently, $H_{\kappa}$ is $\tfrac18$-far from supporting a $4$-wise uniform distribution.
\end{theorem}
\begin{proof}
As a notational shorthand, write $z_{abc}$ for $z_{\{i_a, i_b, i_c\}}$.
Define $\zeta:[\kappa]^6\times \kpm \rightarrow [-5, 5]$ as follows:
\begin{equation}
\begin{multlined}
\zeta(i_1,i_2, i_3, i_4, i_5, i_6, z) = z_{126}z_{134}z_{235}z_{456}
+ z_{256}z_{146}z_{345}z_{123}
+ z_{136}z_{236}z_{145}z_{245}\\
+ z_{124}z_{234}z_{356}z_{156}
+ z_{125}z_{135}z_{346}z_{246}.\label{eq:zetadef}
\end{multlined}
\end{equation}

Observe that for each monomial $z_{T_1}z_{T_2}z_{T_3}z_{T_4}$ of  $\zeta$, for every $j \in [6]$, $\sum_{i=1}^41_{\{T_i \ni j\}} = 2$.
Further, for each $T \subseteq [6]$ with $|T| = 3$, $z_T$ appears exactly once in $\zeta$.
Let $\calZ_6$ be the set of all ordered 6-tuples of distinct elements of $[\kappa]$.  For an ordered tuple $I$, we use $\inT$ to denote membership in $I$.

Define $Q:\kpm \rightarrow \R$ as follows.  Our final polynomial $\calQ$ will be a scaled version of $Q$.
\[
Q(z) = \avg_{I \in \calZ_6}\zeta(I, z).
\]
Observe that $Q$ does not depend on any of $z_{\{1\}}, \ldots z_{\{\kappa\}}$.
By construction, $Q$ contains no constant term, so $\wh{Q}(\emptyset) = 0$.
Clearly $Q(z) \geq -5$ for all $z$ because~\eqref{eq:zetadef} is always at least $-5$.

Now we lower bound the value of $Q$ on all $z$ that satisfy $H_\kappa$. We first show that for any $z'$ that strongly satisfies the Huang predicate, $Q(z') = 5$, then bound $Q(z') - Q(z)$ for any $z$ with Hamming distance at most $\kappa$ from $z'$.
By definition, for each $z'_{T_i}$, we have that $z'_{T_i}\prod_{j \in T_i}z'_j = 1$.
So for each monomial of $Q$,
\begin{align*}
\frac{1}{|\calZ_6|}z'_{T_1}z'_{T_2}z'_{T_3}z'_{T_4} &=\frac{1}{|\calZ_6|}\prod_{i = 1}^4\prod_{j \in T_i}z'_j \\
&= \frac{1}{|\calZ_6|}\prod_{{j \in T_1 \cup T_2 \cup T_3 \cup T_4}}(z'_j)^2 =  \frac{1}{|\calZ_6|},
\end{align*}
where the last line follows from the fact that $\sum_{1=1}^41_{\{T_i \ni j\}} = 2.$
Because there are $5\cdot |\calZ_6|$ monomials in $Q$, their sum is $5$.

Now we consider the case where $z$ does not strongly satisfy the Huang Predicate, but $H_\kappa(z) = 1$.
Any singleton index on which $z$ and $z'$ differ will not change the value of $Q$.
Let $N = \{T : z_T \neq z'_T \}$.
We lower bound $Q$ by counting the number of monomials in which each $z_T$ appears and
\begin{equation*}
Q(z) \geq 5 -\frac2{|\calZ_6|}\sum_{T \in N}\sum_{I \in \calZ}1_{\{\bigwedge_{z_i \in T}i \inT I \}}.
\end{equation*}
For fixed $T$, the number of monomials containing the variables of $z_T$ is
\begin{equation*}
\sum_{I \in \calZ}1_{\{\bigwedge_{z_i \in T}i \inT I \}} = 120(\kappa - 3)(\kappa - 4)(\kappa - 5)\label{eq:permutations}
\end{equation*}
because there are exactly $120$ ways to permute the three indices of $T$ in I and the remaining $\kappa - 3$ indices are permuted in the remaining 3 positions of $I$.
So
\begin{align}
Q(z) &\geq 5 - \frac{240\kappa}{|\calZ_6|}(\kappa-3)(\kappa - 4)(\kappa -5) =5- \frac{240}{(\kappa-1)(\kappa -2)}\label{eq:huang-positive}.
\end{align}

For $\kappa \geq 9$,~\eqref{eq:huang-positive} is at least $ 5 -\tfrac{30}{7}$.  Applying Lemma~\ref{lem:scaleQ}, there exists $\calQ: \kpm \rightarrow \R$ that $\frac18$-separates $H_\kappa$.
\end{proof}
From this and the fact that $\overline{H_\kappa} = 2^{\tilde{O}(k^{1/3})-k}$ (see~\cite{Huang:2013:ARS:2488608.2488666}), we obtain the following corollary.
\begin{corollary}
For sufficiently large $n$ and $k \geq 93$, there exists an efficient algorithm that refutes random instances of $\CSP(H_\kappa)$ with $\wt O(n^2)$ constraints with high probability.
This falsifies Assumption~\ref{assumption:intro-1} in the case of the Huang predicate.
\end{corollary}
\begin{remark}
If we instead choose to scale $Q$ by a factor of $\tfrac15\cdot\frac{\kappa^2-3\kappa+2}{2\kappa^2-6\kappa-44}$ rather than substituting $\kappa = 9$ into~\eqref{eq:huang-positive}, we can achieve a better separation of $\delta = \tfrac{\kappa^2-3\kappa-46}{2\kappa^2-6\kappa-44}$.  For $\kappa \geq 9$, this expression is strictly increasing and it approaches $\tfrac12$ as $\kappa$ grows.
\end{remark}

\subsection{Hamming weight predicates}
The remaining predicates we would like to examine are symmetric, meaning they are functions only of their Hamming weights.  Again for each predicate $P$ we present a multivariate polynomial that $\delta$-separates $P$ for some $0 \leq \delta \leq 1$.  Each of these polynomials can also be written as a univariate polynomial on the Hamming weight of its input, which we will use to show that each of the following polynomials $\delta$-separates its predicate for the appropriate value of $\delta$.
We give the construction below.

\begin{definition}
For $z \in \kpm$ where $z = z_1, \ldots, z_k$,  define $S_z = \sum_{i =1}^kz_i$ and call $S_z$ the Hamming weight of $z$.
\end{definition}
Note that this is analogous to the notion of a Hamming weight of a vector in $\{0, 1\}^k$, but differs in that it is not simply the count of the number of $1$'s.
We define a general predicate that is satisfied when $S_z$ is at least some fixed threshold value $\theta$.
\begin{definition}
For all odd $k$ and any $\theta \in \{-k, -k + 2, \ldots, k-2, k\}$, define the predicate $\thr_k^{\theta}: \kpm \rightarrow \{0, 1\}$ as follows:
\[
\thr_k^\theta(z) = \begin{cases}1 &\textrm{if } S_z \geq \theta \\ 0 &\textrm{otherwise}  \end{cases}
\]
\end{definition}
For example, $\maj_k$ is the same as $\thr_k^{1}$ and $\thr_k^{-k}$ is the trivial predicate satisfied by all $z \in \kpm$.

Because the multilinear separating polynomials we will use are symmetric, we present a transformation to an equivalent univariate polynomial on the Hamming weight of the original input.
\begin{lemma} \label{lem:make-univariate}
Let $\Qdef$ be of the following form for some $a, b, c, d \in \R$:
\begin{equation} \label{eq:deg-4-sym}
Q(z) = a\sum_{\substack{T \subseteq [n]\\ |T| = 1}}z^T+ b\sum_{\substack{T \subseteq [n]\\ |T| = 2}}z^T + c\sum_{\substack{T \subseteq [n]\\ |T| = 3}}z^T + d\sum_{\substack{T \subseteq [n]\\ |T| = 4}}z^T.
\end{equation}
Define $Q_u : \R \rightarrow \R$ as follows:
\begin{equation*} \label{eqn:general-univar}
Q(z) = \frac d{24}S_z^4 + \frac c6S_z^3 + \left(\frac b2 + \frac d{3} - \frac{dk}4 \right)S_z^2 + \left(a+\frac c6\cdot(-3k+2)\right)S_z - \frac{bk}2+ \frac{dk}{24}(3k-6).
\end{equation*}
Then $Q(z) = Q_u(S_z)$ for all $z \in \kpm$.
\end{lemma}
\begin{proof}
We can write~\eqref{eq:deg-4-sym} as follows:
\begin{equation}\label{eq:kraw-multilin}
Q(z) = a\mathscr{K}_1\left(\frac{k-S_z}2;k\right) + b\mathscr{K}_2\left(\frac{k-S_z}2;k\right) + c\mathscr{K}_3\left(\frac{k-S_z}2;k\right)  + d\mathscr{K}_4\left(\frac{k-S_z}2;k\right),
\end{equation}
where $\mathscr{K}_i(\nu; k) = \sum_{j=0}^i(-1)^j{\nu \choose i}{ k-\nu \choose i-j}$ denotes the Krawtchouk polynomial of degree $i$~\cite{krawtchouk, krasikov1996integral}.
Substituting $\nu = \frac{k-S_z}2$, yields the following expressions.  In~\cite{krasikov1996integral} the first three expressions are given explicitly and the fourth can be easily obtained by applying their recursive formula.
\begin{alignat*}{6}
\mathscr{K}_1\left(\frac{k-S_z}2;k\right) &= S_z, &
&\mathscr{K}_3\left(\frac{k-S_z}2;k\right) &&= \frac{S_z^3 - (3k-2)S_z}6,&\\
\mathscr{K}_2\left(\frac{k-S_z}2;k\right) &= \frac{S_z^2 - k}{2},&\qquad
&\mathscr{K}_4\left(\frac{k-S_z}2;k\right) &&= \frac{ S_z^4 + (8-6k)S_z^2  + 3k^2 - 6k}{24}&.\\
\end{alignat*}

Finally, substituting these expressions into~\eqref{eq:kraw-multilin} and by some algebra,
\begin{equation}
Q(z) = \frac d{24}S_z^4 + \frac c6S_z^3 + \left(\frac b2 + \frac d{3} - \frac{dk}4 \right)S_z^2 + \left(a+\frac c6\cdot(-3k+2)\right)S_z - \frac{bk}2+ \frac{dk}{24}(3k-6).\qedhere
\end{equation}
\end{proof}
As a consequence, by choosing values of $a, b, c,$ and $d$, we can work with a univariate polynomial while ensuring that its multivariate analogue is unbiased and has degree at most 4 (degree 3 when $d = 0$).

\subsubsection{Almost-Majority and hardness of learning intersections of halfspaces}
\begin{definition}
Daniely et al.~define the following predicate in order to show hardness of improperly learning intersections of four halfspaces.
\[
I_{8k} = \left( \bigwedge_{i = 0}^3\thr_k^{-1}(z_{ki+1}\ldots z_{ki + k}) \right)\wedge \neg\left(\bigwedge_{i = 4}^7\thr_k^{-1}(z_{ki+1}\ldots z_{ki+k}) 	\right).
\]
\end{definition}
The reduction relies on the assumption that for all $d > 0$, it is hard to distinguish random instances of $\csp(I_{8k})$ with $n^d$ constraints from satisfiable instances.
Because the input variables to each instance of $\thr_k^{-1}$ above are disjoint, it is sufficient to show that each of the first four groups of $k$ variables cannot support a $3$-wise uniform distribution and consequently neither can $I_{8k}$; therefore, from Theorem~\ref{thm:ntwus-refute-2} we deduce that there exists an efficient algorithm that refutes random instances of $\csp(I_{8k})$ with $\wt{O}(n^{3/2})$ constraints with high probability.
Daniely et al.\ define a \pw uniform distribution supported on $I_{8k}$ as well as a \pw uniform distribution supported on $\thr_k^{-1}$, so $t=3$ is optimal.
\begin{theorem}
Assume $k \geq 5$ and $k$ is odd. There exist $\delta = \delta(k) > 0$ where $\delta$ is $\Omega(k^{-4})$ and a degree-$3$ multilinear polynomial $\calQ:\{-1, 1\}^k \rightarrow \R$ that $\delta$-separates $\thr_k^{-1}$.
Consequently, $\thr_k^{-1}$ does not support a $3$-wise uniform distribution.
\end{theorem}
\begin{proof}
Let
\[
Q(z) = (k^2 - k- 1)\sum_{\substack{T \subseteq [n]\\ |T| = 1}}z^T+ (1-k)\sum_{\substack{T \subseteq [n]\\ |T| = 2}}z^T + (1+k)\sum_{\substack{T \subseteq [n]\\ |T| = 3}}z^T
\]
and define $Q_u:\R \rightarrow \R$ as follows:
\begin{align*}
Q_u(s) &=  \frac{1+k}6s^3 + \left(\frac{1-k}2 \right)s^2 + \left(k^2-k-1+\frac{1+k}6\cdot(-3k+2)\right)s - \frac{(1-k)k}2\\
&=  \frac{1+k}6s^3 + \left(\frac{1-k}2 \right)s^2 + \left(\frac{3k^2-7k-4}{6}\right)s - \frac{(1-k)k}2.
\end{align*}
Then by Lemma~\ref{lem:make-univariate}, for all $z \in \kpm$, $Q(z) = Q_u(S_z)$. It  therefore suffices to lower bound $Q_u(s)$ both when $s \geq -1$ and for all $s \in [-k, k]$.

First we show that $Q_u$ is monotonically increasing in $s$.
\begin{align*}
\frac{dQ_u}{ds} &= \tfrac{k+1}{2}s^2 + (1-k)s + \tfrac{3k^2-7k-4}{6}\\\
&= \tfrac 16\left[(k-4)\left(3(s-1)^2 +\tfrac23 + 3k\right) + 15\left(\left(s-\tfrac35\right)^2 + \tfrac{53}{75}\right) \right],
\end{align*}
which is evidently positive for $k \geq 5$.

Because $Q$ is monotonically increasing in $s$, $Q_u(s) \geq Q_u(-k)$ for all $s \in [-k, k]$.
\begin{align}
Q_u(-k) &= \frac{-k-1}6k^3 + \left(\frac{1-k}2 \right)k^2 - \left(\frac{3k^2-7k-4}{6}\right)k - \frac{(1-k)k}2\nonumber\\
&= -\tfrac 16\left[k^4 +7k^3 -13k^2 - k \right]\label{eq:M-neg-thres},\\
&= -\tfrac{1}{6}\left[ k(k-2)(k^2+9k +5) + 9k \right],
\end{align}
which is clearly negative for $k \geq 5$.
Now it just remains to lower-bound $Q_u(s)$ for $s \geq -1$.
Again, since $Q_u$ is monotonically increasing in $s$, we use the value $Q_u(-1)$:
\begin{align*}
Q_u(-1) &= \tfrac{-k-1}6 + \left(\tfrac{1-k}2 \right) - \left(\tfrac{3k^2-7k-4}{6}\right) - \tfrac{(1-k)k}2 = 1.\\
\end{align*}
By applying Lemma~\ref{lem:scaleQ}, there exists an unbiased multilinear polynomial $\calQ: \kpm \rightarrow \R$ of degree $3$ that $\frac{6}{k^4 +7k^3 -13k^2 - k + 6}$-separates $\thr_k^{-1}$.
\end{proof}
Because $\textrm{VAR}_0(I_{8k})$ is evidently $\Omega(k)$ and $\overline{I_{8k}} < \tfrac17$ for all $k \geq 5$, we have the following Corollary.
\begin{corollary}
For odd $k \geq 5$ and sufficiently large $n$, there exists an efficient algorithm that distinguishes between random instances of $\csp(I_{8k})$ with $\wt O(n^{3/2})$ constraints and satisfiable instances with high probability.
\end{corollary}
\begin{remark}
$\thr_3^{-1}$ is the same as $3$-OR and $\thr_5^{-1}$ is the same as is the same as $2$-out-of-$5$-SAT, so this approach can be used to $\Omega_k(1)$-refute $3$-SAT instances and $2$-out-of-$5$-SAT instances with $\wt O_k(n^{3/2})$ constraints, which improves upon
the $O(n^{3/2 + \eps})$ constraints required for refutation of $2$-out-of-$5$-SAT in~\cite{GJ02,GJ03}.
\end{remark}

\subsection{Majority and hardness of approximately agnostically learning halfspaces}
Daniely et al.~show that approximate agnostic improper learning of  halfspaces is hard for all approximation factors $\phi \geq 1$ based on the assumption that for all $d > 0 $ and for sufficiently large odd $k$, it is hard to distinguish between random instances of $\csp(\thr_k^{1})$ with $n^d$ constraints and instances with value at least $1- \frac{1}{10\phi}$.  This is based on the fact that $\max_\calD\E_{z \sim \calD}\left[ \thr_1(z)\right] = 1- \frac{1}{k+1}$, where $\calD$ is a pairwise independent distribution on $\kpm$, and applying SRCSP Assumption~2.  Here we show that for odd $k \geq 25$, $\thr_k^1$ is $0.1$-far from supporting a $4$-wise uniform distribution.  The value $0.1$ is not sharp, but is chosen as a compromise between a reasonably large value and a reasonably simple proof.
\begin{theorem}
There exists a degree-$4$ multilinear polynomial $\calQ: \kpm \rightarrow \R$ that $0.1$-separates $\thr_k^1$ for all odd $k \geq 25$.
\end{theorem}
\begin{proof}
Let
\[
Q(z) = \frac{8}{27\sqrt k}\sum_{\substack{T \subseteq [n]\\ |T| = 1}}z^T - \frac{5}{9k^{3/2}}  \sum_{\substack{T \subseteq [n]\\ |T| = 3}}z^T + \frac{4}{3k^2}\sum_{\substack{T \subseteq [n]\\ |T| = 4}}z^T
\]
and let
\begin{align}
Q_u(s) &= \frac{1}{18k^2}s^4 - \frac 5{54k^{3/2}}s^3 + \left(-\frac 1{3k} + \frac4{9k^2}\right)s^2 + \left(\frac{31}{54\sqrt{k}} - \frac{5}{27k^{3/2}}\right)s + \frac{1}{6}-\frac1{3k}~\label{eq:obnoxious-maj-Q}\\
&= \frac{1}{54}\left[\frac3{k^2}s^4 - \frac5{k^{3/2}}s^3 + \left(-\frac{18}{k} + \frac{24}{k^2} \right)s^2 + \left(\frac{31}{\sqrt{k}} - \frac{10}{k^{3/2}}\right)s + 9 - \frac{18}k \right].\nonumber
\end{align}
Then by Lemma~\ref{lem:make-univariate}, for all $z \in \kpm$, $Q(z) = Q_u(S_z)$.
To simplify $Q$, let $\sigma = sk^{-1/2}$.  Then we can rewrite~\eqref{eq:obnoxious-maj-Q} as follows:
\begin{equation}
Q_u(s) = \tfrac{1}{54}\left[3\sigma^4 - 5\sigma^3 + \left(-{18} + \tfrac{24}{k} \right)\sigma^2 + \left(31 - \tfrac{10}{k}\right)\sigma + 9 - \tfrac{18}k \right]. \label{eq:maj-q-simple}
\end{equation}
First we lower-bound $Q_u(s)$ for all $\sigma \in \R$ using the following expression, which is equivalent to~\eqref{eq:maj-q-simple} by some algebra.
\begin{align*}
Q_u(s)&= \tfrac{1}{54}\left[3(\sigma+\tfrac{29}{18})^2(\sigma - \tfrac{22}{9})^2 + \tfrac{383}{108}\left(\sigma + \tfrac{1832}{1149} \right)^2 - \tfrac{38987378}{837621} +\tfrac{24}{k}\left((\sigma - \tfrac{5}{24})^2 - \tfrac{457}{576}\right)\right]\\
&> \tfrac{1}{54}\left[3(\sigma+\tfrac{29}{18})^2(\sigma - \tfrac{22}{9})^2 + \tfrac{383}{108}\left(\sigma + \tfrac{1832}{1149} \right)^2 - 47 +\tfrac{24}{k}\left(-\tfrac{457}{576}\right)\right] > -\tfrac{48}{54} = -\tfrac{8}{9},
\end{align*}
where the last inequality follows from the fact that $k \geq 24$ and the first two terms are always nonnegative.

Next we lower-bound $Q_u(s)$ for $s > 0$.
\begin{align*}
Q_u(s) &= \tfrac{1}{54}\left[3\sigma^4 - 5\sigma^3 + \left(-{18} + \tfrac{24}{k} \right)\sigma^2 + \left(31 - \tfrac{10}{k}\right)\sigma + 9 - \tfrac{18}k \right]\\
&= \tfrac{1}{54}\left[3(\sigma-\tfrac14)^2(\sigma-\tfrac{25}{12})^2 + \tfrac{41}{120}((\sigma-\tfrac{839}{410})^2+\tfrac{21507}{1344800} + \tfrac{27}{4} + 9\sigma(\sigma-\tfrac{21}{10})^2\right] > \tfrac{1}{8}.
\end{align*}
Applying Lemma~\ref{lem:scaleQ}, there exists $\calQ: \kpm \rightarrow \R$ such that $\calQ$ has degree $4$ and $\calQ$ $\tfrac9{73}$-separates $\thr_k^1$. \qedhere

\end{proof}

\begin{corollary}
For sufficiently large $n$ and $k$, there exists an efficient algorithm that distinguishes between random instances of $\csp(\thr_k^1)$ with $\wt O(n^{2})$ constraints and instances with value at least $0.9$ with high probability.
\end{corollary}

\subsection{Predicates satisfied by strings with Hamming weight at least $-\Theta(\sqrt k)$.}
In light of the fact that the threshold based predicates above are not $4$-wise supporting, one may attempt to find another threshold-based predicate.  Here we show that a symmetric threshold predicate that is $4$-wise supporting must be satisfied by all strings with Hamming weight at least $-\tfrac{\sqrt k}2$.  Furthermore, there exists a symmetric threshold predicate that is $4$-wise supporting with a threshold of $-\Theta(\sqrt k)$ and we sketch its construction.

We also consider the predicate $\thr_k^{-\tfrac12\sqrt k}$.
While it is not used in~\cite{daniely-hardness-of-learning}, we show that it does not support a 4-wise uniform distribution in the interest of obtaining a tighter bound for the Hamming weight above which an unbiased, symmetric predicate is not $4$-wise supporting.
The threshold of $-\tfrac12\sqrt{k}$ is particularly interesting in that it asymptotically matches the threshold $\theta$ below which $\thr_k^{\theta}$ is $4$-wise supporting.
\begin{theorem}
Assume $k \geq 99$ and $k$ is odd.  Then there exists a degree-4 polynomial $\calQ: \kpm \rightarrow \R$ that $\tfrac1{225}$-separates $\thr_k^{-\tfrac12\sqrt k}$.  Consequently, $\thr_k^{-\tfrac12\sqrt{k}}$ is $\tfrac1{255}$-far from $4$-wise supporting.
\end{theorem}
\begin{proof}
Define $\Qdef$ and $\calQ_u:\R \rightarrow \R$ as follows:
\[
Q(z) = \tfrac32k^{-1/2}\sum_{\substack{T \subseteq [n]\\ |T| = 1}}z^T +\tfrac12k^{-1}\sum_{\substack{T \subseteq [n]\\ |T| = 2}}z^T + 2k^{-3/2}\sum_{\substack{T \subseteq [n]\\ |T| = 3}}z^T + 8k^{-2}\sum_{\substack{T \subseteq [n]\\ |T| = 4}}z^T
\]
\[
Q_u(s) = \tfrac{s^4}{3k^2} +\tfrac{s^3}{3k^{3/2}}+\left(-\tfrac{7}{4k} + \tfrac{8}{3k^2} \right)s^2 + \left(\tfrac{1}{2k^{1/2}} +\tfrac{2}{3k^{3/2}} \right)s + \tfrac{3}{4} - \tfrac{2}{k}
\]
Again, for simplicity we set $\sigma = sk^{-1/2}$ and obtain the following expression:
\begin{equation}
Q_u(s) = \tfrac{1}{3}\sigma^4 +\tfrac{1}{3}\sigma^3 -\tfrac74\sigma^2+\tfrac12\sigma+\tfrac34 + \tfrac1k\left(\tfrac83\sigma^2 + \tfrac23\sigma - 2\right).
\end{equation}
Observe that for $k \geq 99$, $\tfrac1k\left(\tfrac83\sigma^2 + \tfrac23\sigma - 2\right) = \tfrac2{3k}\left(\left(2\sigma -\tfrac14\right)^2 - \tfrac{49}{16}\right) >  -\tfrac1{48}$.
We now lower-bound the value of $Q_u$ for $s \geq -\tfrac12k^{1/2}$, or equivalently, $\sigma \geq -\tfrac12$:
\begin{align*}
Q_u(s) &= \tfrac13\sigma^4 + \tfrac13\sigma^3-\tfrac74\sigma^2+\tfrac12\sigma+\tfrac34 +  \tfrac1k\left(\tfrac83\sigma^2 + \tfrac23\sigma - 2\right)\\
&=\tfrac13\left(\sigma - \tfrac{35}{29} \right)^2\left(\sigma + \tfrac12\right)\left(\sigma + \tfrac{200}{69}\right) + \tfrac{61}{12006}\left(\sigma + \tfrac12\right)\left(\left(\sigma+ \tfrac{4631}{3538} \right)^2 +\tfrac{1526073}{12517444}\right)  + \tfrac1k\left(\tfrac83\sigma^2 + \tfrac23\sigma - 2\right) + \tfrac1{24}.\\
\intertext{The first two terms are clearly nonnegative when $\sigma \geq -\tfrac12$, so}
&> \tfrac1{24} - \tfrac1{48} = \tfrac1{48}.
\end{align*}
We also show that $Q_u(s) \geq -\tfrac{14}{3}$ for all $s \in \R$.
\begin{align*}
Q_u(s) &= \tfrac13\sigma^4 + \tfrac13\sigma^3-\tfrac74\sigma^2+\tfrac12\sigma+\tfrac34 +  \tfrac1k\left(\tfrac83\sigma^2 + \tfrac23\sigma - 2\right)\\
&= \left(\sigma+\tfrac{19}9\right)^2\left(\sigma -\tfrac{29}{18} \right)^2 + \tfrac{211}{486}\left(\sigma + \tfrac{397}{211}\right)^2 +\tfrac{195823}{8306226} + \tfrac1k\left(\tfrac83\sigma^2 + \tfrac23\sigma - 2\right) - \tfrac{14}3\\
&\geq \left(\sigma+\tfrac{19}9\right)^2\left(\sigma -\tfrac{29}{18} \right)^2 + \tfrac{211}{486}\left(\sigma + \tfrac{397}{211}\right)^2 + \tfrac{1079081}{365473944}  -\tfrac{14}3.
\end{align*}
The first three terms are always nonnegative, so $Q_u(s) \geq -\tfrac{14}3$.

Applying Lemma~\ref{lem:scaleQ}, $T_k^{-\tfrac12\sqrt{k}}$ is $\tfrac1{255}$-far from supporting a $4$-wise uniform distribution.
\end{proof}
Now we demonstrate that there exists a $4$-wise uniform distribution supported on $\thr_k^{1 - 2\sqrt{k+1}}$ when $k = 2^m - 1$ for some integer $m \geq 3$.
\begin{claim}
Assume $k = 2^m - 1$ for some integer $m \geq 3$.  Then there exists a $4$-wise uniform distribution supported only on $z \in \kpm$ such that $S_z \geq 1-2\sqrt{k+1}$.
\end{claim}
\begin{proof}
Let $\calC$ be a binary BCH code of length $k$ with designed distance $2\iota + 1$ and let $\calC^{\bot}$ be its dual.
Then the uniform distribution on the codewords of $\calC$ is $2\iota$-wise uniform~\cite{alon1986fast, MS77}; see also~\cite[Ch 16.2]{alon2004probabilistic}.

Let $c = c_1\ldots c_k$  be a codeword of $\calC^{\bot},$ where each $c_i \in \{-1, 1\}$.  The Carlitz-Uchiyama bound~\cite[page 280]{MS77} states that for all $c \in \calC^{\bot}$,
\[
\sum_{i = 1}^{k}\tfrac12(1 - c_i) \leq \tfrac{k + 1}2 + (\iota - 1)\sqrt{k + 1}.
\]
Observe that the quantity $\tfrac12(1 - c_i)$ simply maps $c_i$ from $\{-1, 1\}$ to $\{0, 1\}$ so that we can write the bound to match the presentation in~\cite{MS77}.
Therefore,
\begin{align*}
S_c &= \sum_{i = 1}^{k}c_i\\
&= k - 2\sum_{i = 1}^{k}\tfrac12(1-c_i)\\
&\geq k - (k + 1) - (2\iota - 2)\sqrt{k + 1}\\
&= -1 -(2\iota - 2)\sqrt{k + 1}.
\end{align*}
Setting $\iota = 2$, we can obtain $4$-wise uniformity on this distribution and each string in the support of the distribution has Hamming weight at least $-1 -2\sqrt{k +1}.$
\end{proof}
\begin{remark}
In order to construct a $4$-wise uniform distribution for any value of $k$, one could simply express $k$ as a sum of powers of 2, construct separate distributions on disjoint variables as described above for each power of 2 (down to the minimum length for which we can achieve distance at least $5$, after which point we use the uniform distribution, and obtain a $4$-wise uniform distribution.  The total Hamming weight of a vector supported by this distribution would then be at least $-O(\sqrt k)$.
\end{remark}
%
%
%
%
%


\section{SOS refutation proofs}
\label{sec:sos}
\subsection{The SOS proof system}
We first define the SOS proof system introduced in \cite{GV01}.  Call a polynomial $q \in \R[X_1,\ldots,X_n]$ sum-of-squares (SOS) if there exist $q_1, \ldots, q_{\ell} \in \R[X_1,\ldots,X_n]$ such that $q = q_1^2 + \cdots + q_{\ell}^2$.
\begin{definition}
Let $X=(X_1,\ldots,X_n)$ be indeterminates.  Let $q_1,\ldots, q_m,r_1,\ldots,r_{m'} \in \R[X]$ and let $A = \{q_1 \geq 0, \ldots, q_m \geq 0\} \cup \{r_1 = 0, \ldots, r_{m'}\}$.  There is a degree-$d$ SOS proof that $A$ implies $s \geq 0$, written as
\[
A \vdash_d s \geq 0,
\]
if there exist SOS $u_0,u_1, \ldots, u_m \in \R[X]$ and $v_1,\ldots,v_{m'} \in \R[X]$ such that
\[
s = u_0 + \sum_{i = 1}^m u_i q_i + \sum_{i = 1}^{m'} v_i r_i
\]
with $\deg(u_0) \leq d$, $\deg(u_i q_i) \leq d$ for all $i \in [m]$, and $\deg(v_i r_i) \leq d$ for all $i \in [m']$.  If it also holds that $u_0,u_1,\ldots,u_{m} = 0$, we will write $A \vdash_d s = 0$.
\end{definition}
It is well-known that a degree-$d$ SOS proof can be found using an SDP of size $n^{O(d)}$ if it exists \cite{Sho87, Par00, Las00, Las01}.

In this section, we will take the set $A$ to be $\{x_i^2 = 1\}_{i \in [n]}$, enforcing that variables are $\pm 1$-valued.  We show that with high probability there exists a low-degree SOS proof that a polynomial representing the value of a CSP instance is close to its expectation.

For more information on the SOS proof system and its applications to approximation algorithms, see, e.g., \cite{OZ13, Lau09}.

\subsection{SOS certification of quasirandomness}
All of our SOS results rely on the following theorem, which is the SOS version of Theorem~\ref{thm:main}.
\begin{theorem} \label{thm:main-sos}
For $k \geq 2$ and $p \geq n^{-k/2}$, let $\{w(T)\}_{T \in [n]^k}$ be independent random variables such that for each $T \in [n]^k$,
\begin{align}
\E[w(T)] &= 0  \\
\Pr[w(T) \ne 0] &\leq p  \\
\abs{w(T)} &\leq 1.
\end{align}
Then, with high probability, 
\[
\{x_i^2 \leq 1\}_{i \in [n]} \vdash_{2k} \sum_{T \in [n]^k} w(T) x^T \leq 2^{O(k)} \sqrt{p} n^{3k/4} \log^{3/2} n.
\]
\end{theorem}
This theorem was essentially proven by Barak and Moitra \cite{BM15}.  We give a proof in Appendix~\ref{subsec:main-sos}.  We first use this theorem to show that an SOS version of Lemma~\ref{lem:smaller-monomials} holds.
\begin{lemma}
\label{lem:smaller-monomials-sos}
Let $S \subseteq [k]$ with $|S| = s > 0$.  Let $\tau \in \N$ and let $\{w_U(i)\}_{U \in [n]^s, i \in [\tau]}$ be independent random variables satisfying conditions \eqref{eq:w-mean-0}, \eqref{eq:w-mostly-0}, and \eqref{eq:w-bounded} for some $p \geq \frac{1}{\tau n^{s/2}}$.  Then, with high probability,
\[
\{x_i^2 \leq 1\}_{i \in [n]} \vdash_{2s} \sum_{U \in [n]^s} x^U \sum_{j = 1}^{\tau} w_U(j) \leq
\begin{cases}
2^{O(s)} \sqrt{\tau p} \cdot n^{3s/4} \log^{5/2} n \quad & \text{if $s \geq 2$} \\
4 \max\{\sqrt{\tau p}, 1\} \cdot n \log n \quad & \text{if $s = 1$.}
\end{cases}
\]
\end{lemma}

\begin{proof}
We sketch the differences from the proof of Lemma~\ref{lem:smaller-monomials} given in Section~\ref{sec:smaller-monomials}.  For $s \geq 2$, the lemma follows by using Theorem~\ref{thm:main-sos} instead of Theorem~\ref{thm:main}.  If $s = 1$, it suffices to show that
\[
\{x_i^2 \leq 1\} \vdash_2 v(i)x_i \leq \abs{v(i)}.
\]
for any $v$ since summing over all $i$ as in \eqref{eq:1-case} finishes the proof.  If $v_i \geq 0$, observe that
\[
\abs{v_i} - v(i) x_i = \frac{\abs{v(i)}}{2}(x_i-1)^2 + \frac{\abs{v(i)}}{2}(1-x_i^2).
\]
If $v(i) < 0$, we use $(x_i + 1)^2$ instead of $(x_i-1)^2$.
\end{proof}
The lemma implies an SOS version of Lemma~\ref{lem:fourier-refute}.  To make this precise, we define a specific polynomial representation of $\widehat{D_{\calI,x}}(S)$:
\[
\widehat{D_{\calI,x}}(S)^{\mathrm{poly}} = \frac{1}{m} \sum_{T \in [n]^k} \sum_{c \in \{\pm 1\}^k} \indic{\{(T,c) \in \calI\}} c^S x^S_T,
\]
where $x_T^S = \prod_{i \in S} x_{T_i}$.  Note that this is a polynomial in the $x_i$'s.

We can show these polynomials are not too large.
\begin{lemma}
\label{lem:fourier-refute-sos}
Let $\emptyset \ne S \subseteq [k]$ with $|S| = s$.  Then
\begin{align*}
\{x_i^2 \leq 1\}_{i \in [n]} &\vdash_{2s} \widehat{D_{\calI,x}}(S)^{\mathrm{poly}} \leq \frac{2^{O(k)} \max\{n^{s/4}, \sqrt{n}\} \log^{5/2} n}{\expm^{1/2}} \\
\{x_i^2 \leq 1\}_{i \in [n]} &\vdash_{2s} \widehat{D_{\calI,x}}(S)^{\mathrm{poly}} \geq -\frac{2^{O(k)} \max\{n^{s/4}, \sqrt{n}\} \log^{5/2} n}{\expm^{1/2}}.
\end{align*}
with high probability, assuming also that $\expm \geq \max\{n^{s/2}, n\}$.
\end{lemma}
\begin{proof}
The proof is essentially the same as that of Lemma~\ref{lem:fourier-refute}. The expression we bound in that proof is exactly $\widehat{D_{\calI,x}}(S)^{\mathrm{poly}}$. We use Lemma~\ref{lem:smaller-monomials-sos} instead of Lemma~\ref{lem:smaller-monomials} to show that this can be done in degree-$2s$ SOS.
\end{proof}
Based on Lemma~\ref{lem:AGM}, we will think of Lemma~\ref{lem:fourier-refute-sos} as giving an SOS proof of quasirandomness.  Below, we use it to prove SOS versions of Theorems~\ref{thm:strong-ref-2} and~\ref{thm:ntwus-refute-2}.

\subsection{Strong refutation of any $k$-CSP}
We now define the natural polynomial representation of $\mathrm{Val}_{\calI}(x)$:
\[
\mathrm{Val}_{\calI}^{\mathrm{poly}}(x) = \frac{1}{m} \sum_{T \in [n]^k} \sum_{c \in \{\pm 1\}^k} \indic{\{(T,c) \in \calI\}} \left(\sum_{S \subseteq [k]} \wh{P}(S) c^S x_T^S \right),
\]
where $x_T^S$ is as above.

We can then give an SOS proof strongly refuting CSP($P$).
\begin{theorem} \label{thm:strong-ref-2-SOS}
Given an instance $\calI \sim \calF_P(n,p)$ of $\CSP(P)$,
\[
\{x_i^2 \leq 1\}_{i \in [n]} \vdash_{2k} \mathrm{Val}^{\mathrm{poly}}_{\calI}(x) \leq \expP + \gamma
\]
with high probability when $\expm \geq \frac{2^{O(k)} n^{k/2} \log^{5} n}{\gamma^2}$.
\end{theorem}
\begin{proof}
By rearranging terms, we see that
\[
\mathrm{Val}^{\mathrm{poly}}_{\calI}(x) = \expP + \sum_{\emptyset \ne S \subseteq [k]} \wh{P}(S) \widehat{D_{\calI,x}}(S)^{\mathrm{poly}}.
\]
Note that this is just Plancherel's Theorem in SOS.  The theorem then follows from Lemma~\ref{lem:fourier-refute-sos} and the observation that $\sum_{S \subseteq [k]} |\wh{P}(S)| \leq 2^{O(k)}$.
\end{proof}

\subsection{$\Omega(1)$-refutation of non-$t$-wise supporting CSPs}
\label{sec:sos-ntwus}
\begin{theorem}
\label{thm:ntwus-refute-2-sos}
Let $P$ be $\delta$-far from being \twus.  Given an instance $\calI \sim \calF_P(n,p)$ of $\CSP(P)$,  
\[
\{x_i^2 = 1\}_{i \in [n]} \vdash_{\max\{k,2t\}} \mathrm{Val}^{\mathrm{poly}}_{\calI}(x) \leq 1 - \delta + \gamma.
\]
with high probability when $\expm \geq \frac{2^{O(k)} n^{t/2} \log^{5} n}{\gamma^2}$ and $t \geq 2$.
\end{theorem}
To prove this theorem, we will need to following claim, which says that any true inequality in $k$ variables over $\{-1,1\}^k$ can be proved in degree-$k$ SOS.  Recall that the multilinearization of a monomial $z_1^{s_1}z_2^{s_2} \cdots z_k^{s_k} \in \R[z_1,\ldots,z_k]$ is defined to be $z_1^{s_1 \bmod 2}z_2^{s_2 \bmod 2} \cdots z_k^{s_k \bmod 2}$, i.e., we replace all $z_i^2$ factors by $1$.  We extend this definition to all polynomials in $\R[z_1,\ldots,z_k]$ by linearity.
\begin{claim}
\label{cl:deg-k-ineq}
Let $f:\{-1,1\}^k \to \R$ such that $f(z) \geq 0$ for all $z \in \{-1,1\}^k$ and let $f^{\mathrm{poly}}$ be the unique multilinear polynomial such that $f(z) = f^{\mathrm{poly}}(z)$ for all $z \in \{-1,1\}^k$.  Then
\[
\{z_i^2 = 1\}_{i \in [k]} \vdash_k f^{\mathrm{poly}}(z) \geq 0.
\]
\end{claim}
\begin{proof}
Since $f(x) \geq 0$, there exists a Boolean function $g:\{-1,1\}^k \to \R$ such that $g(z)^2 = f(z)$ for all $z \in \{-1,1\}^k$.  Let $g^{\mathrm{poly}}$ be the unique multilinear polynomial such that $g(z) = g^{\mathrm{poly}}(z)$ for all $z \in \{-1,1\}^k$.  Since $g^{\mathrm{poly}}(z)^2 = f^{\mathrm{poly}}(z)$ for all $z \in \{-1,1\}^k$, uniqueness of the multilinear polynomial representation of $f$ implies that the multilinearization of $(g^{\mathrm{poly}})^2$ is equal to $f^{\mathrm{poly}}$.  Written another way, we have that $\{z_i^2 = 1\}_{i \in [k]} \vdash_k f^{\mathrm{poly}}(z) = g^{\mathrm{poly}}(z)^2$.  This implies that $\{z_i^2 = 1\}_{i \in [k]} \vdash_k f^{\mathrm{poly}}(z) \geq 0$.
\end{proof}

\begin{proof}[Proof of Theorem~\ref{thm:ntwus-refute-2-sos}]
The proof is an SOS version of Proof~2 of Theorem~\ref{thm:ntwus-refute-2} above.  Claim~\ref{cl:deg-k-ineq} implies that for $Q$ of degree at most $t$ that $\delta$-separates $P$,
\[
\{z_i^2 = 1\}_{i \in [k]} \vdash_k Q(z)-P(z)+1-\delta \geq 0.
\]
Summing over all constraints, we get
\[
\{x_i^2 = 1\}_{i \in [n]} \vdash_{k} \mathrm{Val}^{\mathrm{poly}}_{\calI}(x) - (1-\delta) \leq \frac{1}{m} \sum_{T \in [n]^k} \sum_{c \in \{\pm 1\}^k} \indic{\{(T,c) \in \calI\}} \left(\sum_{S \subseteq [k]} \wh{Q}(S) c^S x_T^S \right),
\]
Rearranging terms as in the proof of Theorem~\ref{thm:strong-ref-2-SOS}, we see that the right hand side is equal to
\[
\sum_{S \subseteq [k]} \wh{Q}(S) \widehat{D_{\calI,x}}(S)^{\mathrm{poly}}.
\]
Since $Q$ has mean $0$, $|Q| \leq 2^k$ and $\sum_{S \subseteq [k]} |\wh{P}(S)| \leq 2^{O(k)}$.  The theorem then follows from Lemma~\ref{lem:fourier-refute-sos}.
\end{proof}
With Corollary~\ref{cor:lp-granularity}, the theorem implies that we can $\Omega_k(1)$-refute any $\CSP(P)$ in SOS when $P$ is not $t$-wise supporting.
\begin{corollary}
\label{cor:no-t-wise-ref-2-SOS}
Let $P$ be a predicate that does not support any $t$-wise uniform distribution.  Given an instance $\calI \sim \calF_P(n,p)$ of $\CSP(P)$,
\[
\{x_i^2 = 1\}_{i \in [n]} \vdash_{\max\{k,2t\}} \mathrm{Val}_{\calI}(x) \leq 1 - 2^{-\wt{O}(k^t)} + \gamma
\]
with high probability when $\expm \geq 2^{\wt{O}(k^t)} n^{t/2} \log^{5} n$ and $t \geq 2$.
\end{corollary}


\section{Directions for future work} \label{ref:conclusion}

It would be interesting to show analogous efficient refutation results for models of random $\CSP(P)$ in which literals are not used. This would allow for results on, say, refuting $q$-colorability for random $k$-uniform hypergraphs.  We give a simple result on refuting $q$-colorability of random hypergraphs in Appendix~\ref{sec:hypergraphs}, but it follows from refutation of binary CSPs and perhaps a stronger result could be proven by studying CSPs with larger alphabets.  For some predicates (e.g., monotone Boolean predicates), random CSP instances are trivially satisfiable when there are no literals.  However for such predicates one could consider a ``Goldreich~\cite{Gol00}-style'' model in which each constraint is randomly either~$P$ or~$\neg P$ applied to~$k$ random variables.

Additionally, it would be good to investigate whether our refutation algorithms can be extended from the purely random CSP$(P)$ setting to the ``smoothed''/``semi-random'' setting of Feige~\cite{Fei07}, in which the $m$ constraints scopes are worst-case and only the negation pattern for literals is random.  Feige showed how to efficiently refute random $3$-SAT instances with $m \geq \wt{O}(n^{3/2})$ constraints even in this model.

Another valuable open direction would be to shore up the known proof-complexity evidence suggesting that $\wt{\Theta}(n^{t/2})$ constraints might be necessary to refute random $\CSP(P)$ when $P$ is not \twus.  The natural question here is whether the SOS lower bound of~\cite{BCK15} can be extended from non-pairwise uniform supporting and $m = O(n)$ constraints, to non-$t$-wise uniform supporting and $m = O(n^{t/2-\eps})$ constraints.  (Of course, it would also be good to eliminate the pruning step from their random instances.)  One might also investigate the more refined question of whether, for $P$ that are $\delta$-far from \twus, one can improve on $\delta$-refutation when there are $m \geq \wt{O}(n^{t/2})$ constraints.

Followup work on the very interesting paper~\cite{FKO06} of Feige, Kim, and Ofek also seems warranted.  Recall that it gives a \emph{nondeterministic} refutation algorithm for random $3$-SAT when $m \geq O(n^{1.4})$ (as well as a subexponential-time deterministic algorithm).  This raises the question of whether there exist polynomial-size refutations for random CSP$(P)$ instances that are nevertheless hard to find efficiently.

Finally, we suggest trying to rehabilitate the hardness-of-learning results in~\cite{daniely-hardness-of-learning}, given our new knowledge about what random $\CSP(P)$ instances seem hard to refute.  As mentioned, the followup work of Daniely and Shalev-Shwartz~\cite{2014arXiv1404.3378D} shows hardness of PAC-learning DNFs with $\omega(\log n)$ terms based on the very reasonable assumption that refuting random $k$-SAT requires $n^{f(k)}$ constraints for some $f(k) = \omega(1)$.  Subsequent work by Daniely~\cite{Dan15} shows hardness of approximately PAC-learning halfspaces assuming that refuting random $k$-XOR is hard both when $m = n^{c\sqrt{k}\log k}$ and when $k$ is polylogarithmic in $n$ and $m = n^{ck}$ for some $c >0$.  One future direction would be to obtain hardness results for agnostically learning intersections of halfspaces.

\subsection*{Acknowledgments}
The authors would like to thank Amin Coja--Oghlan for help with the literature, and Boaz Barak and Ankur Moitra for permission to reprint the proof of the strong $k$-XOR refutation result.  The last author would like to thank Anupam Gupta for several helpful discussions. 

\bibliographystyle{alpha}
\bibliography{witmer,sarahrefs,odonnell-bib}
\newpage

\appendix

\section{Proof of Theorem~\ref{thm:main}}
\label{sec:xor}
We restate Theorem~\ref{thm:main}:
\begin{customthm}{\ref{thm:main}}
For $k \geq 2$ and $p \geq n^{-k/2}$, let $\{w(T)\}_{T \in [n]^k}$ be independent random variables such that for each $T \in [n]^k$,
\begin{align}
\E[w(T)] &= 0  \\
\Pr[w(T) \ne 0] &\leq p  \\
\abs{w(T)} &\leq 1.
\end{align}
Then there is an efficient algorithm certifying that
\begin{equation}
\label{eq:main-thm}
\sum_{T \in [n]^k} w(T) x^T \leq  2^{O(k)} \sqrt{p} n^{3k/4} \log^{3/2} n.
\end{equation}
for all $x \in \R^n$ with $\norm{x}_{\infty} \leq 1$ with high probability.
\end{customthm}

The proof of this theorem constitutes the remainder of this section.  It will often be convenient to consider $T \in [n]^k$ to be $(T_1,T_2) \in [n]^{k_1} \times [n]^{k_2}$ with $k_1 + k_2 = k$. In such situations, we will write $w(T) = w(T_1,T_2)$. For intuition, the reader can think of the special case of $w(T) \in \{-1,0,1\}$ for all $T$ and $y \in \{-1,1\}^n$.  Under these additional constraints, $\sum_{T \in [n]^k} w(T) x^T$ is $\val(\calI) - \frac{1}{2}$ for a random $k$-XOR instance $\calI$ so we are certifying that a random $k$-XOR instance does not have value much bigger than $\frac{1}{2}$.

\subsection{The even arity case}
When $k$ is even, we can think of $\sum_{T \in [n]^k} w(T) x^T$ as a quadratic form:
\begin{equation}
\label{eq:even-k}
\sum_{T \in [n]^k} w(T) x^T = \sum_{T_1,T_2 \in [n]^{k/2}} w(T_1,T_2) y_{T_1} y_{T_2},
\end{equation}
where $y_U = x^U$.  We give two methods to certify that the value of this quadratic form is at most $O_k(\sqrt{p} n^{3k/4} \log{n})$.  The first method uses an SDP-based approximation algorithm and works only for $x \in \{-1,1\}^n$.  The second method uses ideas from random matrix theory and works for any $x$ with $\norm{x}_{\infty} \leq 1$.

\paragraph{Approximation algorithms approach} If $x \in \{-1,1\}^n$, we can apply an approximation algorithm of Charikar and Wirth \cite{CW04} for quadratic programming.  They prove the following theorem:
\begin{theorem}{\textup{\cite[Theorem 1]{CW04}}}
\label{thm:CW}
Let $M$ be any $n \times n$ matrix with all diagonal elements $0$.  There exists an efficient randomized algorithm that finds $y \in \{-1,1\}^n$ such that
\[
\E[y^{\top} M y] \geq \Omega\left(\frac{1}{\log n}\right) \max_{x \in \{-1,1\}^n} x^{\top} M  x.
\]
\end{theorem}
By Markov's Inequality, this statement holds with probability at least $1/2$.  We can run the algorithm $O(\log n)$ times to get a high probability result.  To apply Theorem~\ref{thm:CW}, we separate out the diagonal terms of \eqref{eq:even-k}, rewriting it as
\begin{equation}
\label{eq:even-k-2-parts}
\sum_{T_1 \ne T_2 \in [n]^{k/2}} w(T_1,T_2) y_{T_1} y_{T_2} + \sum_{U \in [n]^{k/2}} w(U,U) y_U^2.
\end{equation}
We can certify that each of the two terms in this expression is at most $O(\sqrt{p} n^{3k/4} \log n)$.  For the first term, we will need the following claim.
\begin{claim} \label{cl:approx-alg-term1}
With high probability, it holds that
\[
\sum_{T_1, T_2 \in [n]^{k/2}} w(T_1,T_2) y_{T_1} y_{T_2} \leq O(\sqrt{p}n^{3k/4}).
\]
\end{claim}
This follows from applying Bernstein's Inequality (Theorem~\ref{thm:bernstein}) for fixed $y$ and then taking a union bound over all $y \in \{-1,1\}^{n^{k/2}}$.  Using the claim, we see that Theorem~\ref{thm:CW} allows us to certify that the value of the first term in \eqref{eq:even-k-2-parts} is at most $O(\sqrt{p} n^{3k/4} \log n)$.

We will use the next claim to bound the second term of \eqref{eq:even-k-2-parts}.
\begin{claim} \label{cl:approx-alg-term2}
With high probability, it holds that
\[
\sum_{U \in [n]^{k/2}} |w_{U,U}| \leq O(\sqrt{p} n^{3k/4}).
\]
\end{claim}
Since the $|w_T| \leq 1$ and $\Pr[w_T \ne 0] \leq p$, the claim follows from the Chernoff Bound.  The second term of \eqref{eq:even-k-2-parts} is upper bounded by $\sum_{U \in [n]^{k/2}} |w_{U,U}|$ and we can compute this quantity in polynomial time to certify that its value is at most $O(\sqrt{p} n^{3k/4})$.

\paragraph{Random matrix approach} Observe that \eqref{eq:even-k} is $y^{\top} B y$ for a matrix $B$  indexed by $U \in [n]^{k/2}$ so that $B_{U_1, U_2} = w(U_1, U_2)$.  Then $y^{\top} B y \leq \norm{B} \norm{y}^2$.  To certify that $y^{\top} B y$ is small, we compute $\norm{B}$.  We need to show that $\norm{B}$ is small with high probability.  First, note that $\norm{B}$ is equal to the norm of the $2n^{k/2} \times 2n^{k/2}$ symmetric matrix
\[
\tilde{B} = \begin{pmatrix}
0 & B \\
B^{\top} & 0
\end{pmatrix}
\]
For example, this appears as (2.80) in \cite{Tao12}.  The upper triangular entries of $\tilde{B}$ are independent random variables with mean $0$ and variance at most $p$ by the properties of the $w_S$'s.  We can then apply a standard bound on the norm of random symmetric matrices \cite{Tao12}.
\begin{proposition}
\label{prop:random-matrix-norm}
\textup{\cite[Proposition 2.3.13]{Tao12}}
Let $M$ be a random symmetric matrix $n \times n$ whose upper triangular entries $M_{ij}$ with $i \geq j$ are independent random variables with mean $0$, variance at most $1$, and magnitude at most $K$.  Then, with high probability,
\[
\norm{M} = O(\sqrt{n} \log{n} \cdot \max\{1,K/\sqrt{n}\}).
\]
\end{proposition}
Let $\tilde{B}'  = \frac{1}{\sqrt{p}} \tilde{B}$.  The upper triangular entries of $\tilde{B}'$ are independent random variables with mean $0$, variance at most $1$, and magnitude at most $1/\sqrt{p}$.  Applying Proposition~\ref{prop:random-matrix-norm} to $\tilde{B}'$ shows that $\|B\| = O\left(kn^{k/4} \sqrt{p} \log{n} \cdot \max\left\{1,\frac{1}{\sqrt{p}n^{k/4}}\right\}\right)$ with high probability.  Since $\norm{y}_{\infty} \leq 1$ by assumption, $\norm{y}^2 \leq n^{k/2}$ and \eqref{eq:even-k} is at most $O(k \sqrt{p} n^{3k/4} \log{n})$ with high probability when $p \geq n^{-k/2}$.

\subsection{The odd arity case}
Fix an assignment $x \in [-1,1]^n$.  For $i \in [n]$, the monomials containing $x_i$ can contribute at most $W_i := \abs{\sum_{T \in [n]^{k-1}} w(T,i) x^T}$ to the objective if $x_i$ is set optimally. By Cauchy-Schwarz,
\begin{equation}
\label{eq:normal-cs}
\sum_{T \in [n]^k} w(T) x^T \leq  \sum_{i \in [n]} W_i \leq \sqrt{n} \sqrt{\sum_{i \in [n]} W_i^2},
\end{equation}
so it suffices to bound $\sum_{i \in [n]} W_i^2$.  We will write this as a quadratic polynomial and then bound it using spectral methods:
\begin{align}
\sum_{i \in [n]} W_i^2 &= \sum_{T,U \in [n]^{k-1}} \sum_{i \in [n]} w(T, i) w(U,i) x^{T} x^{U} \nonumber \\
&=  \sum_{T_1',T_2',U_1',U_2' \in [n]^{\frac{k-1}{2}}} \sum_{i \in [n]} w(T_1', U_1', i) w(T_2', U_2', i) x^{(T_1',T_2')} x^{(U'_1,U'_2)}. \label{eq:shuffle}
\end{align}
Define the $n^{k-1} \times n^{k-1}$ matrix $A$ indexed by $[n]^{k-1}$:
\begin{equation}
\label{eq:A-def}
A_{(i_1, i_2), (j_1,j_2)} = \begin{cases}
\sum_{\ell \in [n]} w(i_1,j_1,\ell)  w(i_2, j_2, \ell) \quad & \text{if $(i_1,j_1) \ne (i_2, j_2)$} \\
0 \quad & \quad \text{otherwise},
\end{cases}
\end{equation}
where we have divided the indices of $A$ into 2 blocks of $\frac{k-1}{2}$ coordinates each.  Define $x^{\otimes k-1} \in \R^{[n]^{k-1}}$ so that $x^{\otimes k-1} (T) = x^T$.  Then \eqref{eq:shuffle} is equal to
\begin{equation}
\label{eq:break-up-sum}
(x^{\otimes k-1})^{\top} A x^{\otimes k-1} + \sum_{T,U \in [n]^{\frac{k-1}{2}}} (x^T)^2 (x^U)^2  \sum_{i \in [n]} w(T, U, i)^2.
\end{equation}
The first term is at most $\norm{A}n^{k-1}$ since the variables are bounded.  We can compute $\norm{A}$ to certify this.  With high probability, $\norm{A}$ is not too big.
\begin{lemma}
\label{lem:norm-bound}
Let $k \geq 3$ and $p \geq n^{-k/2}$.  Let $\{w(T)\}_{T \in [n]^k}$ be indepedent random variables satisfying conditions \eqref{eq:w-mean-0}, \eqref{eq:w-mostly-0}, and \eqref{eq:w-bounded} above.  Let $A$ be defined as in \eqref{eq:A-def}.  With high probability,
\[
\norm{A} \leq 2^{O(k)} p n^{k/2} \log^3 n.
\]
\end{lemma}
We can therefore certify that the first term is $2^{O(k)} p n^{3k/2 - 1} \log^3 n$.  We will prove the lemma in Appendix~\ref{sec:norm-bound-pf}.

The second term of \eqref{eq:break-up-sum} is at most $\sum_{T \in [n]^k} w(T)^2$.  We can easily compute this and the Chernoff Bound implies that its value is at most $p n^{3k/2 - 1}$ with high probability.

So far, with high probability we can certify that $\sum_{i \in [n]} W_i^2 = 2^{O(k)} p n^{3k/2 - 1} \log^3 n$.   Plugging this bound into \eqref{eq:normal-cs} concludes the proof.

\begin{remark}
It would have been more natural to have written $\sum_{i \in [n]} W_i^2 = (x^{\otimes k-1})^{\top} A' x^{\otimes k-1}$ for $A'$ such that $A'_{T,U} = \sum_{i \in [n]} w(T, i) w(U,i)$.  However, $\norm{A'}$ could be too large because of the contribution of the second term in \eqref{eq:break-up-sum}.  We use the additional assumption that $\norm{x}_{\infty} \leq 1$ to get around this issue.
\end{remark}

\begin{remark}
We have defined $A$ so that $A_{(i_1,i_2), (j_1,j_2)} = \sum_{\ell \in [n]} w(i_1,j_1,\ell)  w(i_2, j_2, \ell)$, not $A_{i, j} = \sum_{\ell \in [n]} w(i,\ell) w(j,\ell)$.  The reduces the correlation among entries $w(b,c)$ and $w(b,c')$ for $c \ne c'$.  Intuitively, $A$ looks more like a random matrix with independent entries, so we can bound its norm using the trace method.  See the proof of Lemma~\ref{lem:norm-bound} in Appendix~\ref{sec:norm-bound-pf}.
\end{remark}

\subsection{An SOS version}
\label{subsec:main-sos}
In this section, we will prove the SOS version of Theorem~\ref{thm:main}.
\begin{customthm}{\ref{thm:main-sos}}
For $k \geq 2$ and $p \geq n^{-k/2}$, let $\{w(T)\}_{T \in [n]^k}$ be independent random variables such that for each $T \in [n]^k$,
\begin{align*}
\E[w(T)] &= 0  \\
\Pr[w(T) \ne 0] &\leq p  \\
\abs{w(T)} &\leq 1.
\end{align*}
Then, with high probability, 
\[
\{x_i^2 \leq 1\}_{i \in [n]} \vdash_{2k} \sum_{T \in [n]^k} w(T) x^T \leq 2^{O(k)} \sqrt{p} n^{3k/4} \log^{3/2} n.
\]
\end{customthm}
Rather than writing out the full proof, we will indicate the small changes required to convert the above proof of Theorem~\ref{thm:main} into SOS form.
\paragraph{Even arity} The random matrix proof for the even case can easily be converted into an SOS proof with degree $k$.  When $O(k \sqrt{p} n^{k/4} \log{n}) I - B \succeq \mathbf{0}$, there exists a matrix $M$ such that $M^{\top}M = O(k \sqrt{p} n^{k/4} \log{n}) I - B$. Then
\[
O(k \sqrt{p} n^{k/4} \log{n}) \norm{y}^2 - y^{\top}By = (My)^{\top}(My) = \sum_{T \in [n]^{k/2}} \left(\sum_{U \in [n]^{k/2}} M_{T,U} y_{U}\right)^2
\]
so
\[
\{x_i^2 \leq 1\}_{i \in [n]} \vdash_k \sum_{T \in [n]^k} w(T) x^T \leq O(k \sqrt{p} n^{3k/4} \log{n}).
\]

\paragraph{Odd arity} A couple of additional issues arise in the odd case.  First of all, the square root in \eqref{eq:normal-cs} is not easily expressed in SOS, so we instead prove the squared version
\begin{equation}
\label{eq:xor-squared}
\left(\sum_{T \in [n]^k} w(T) x^T\right)^2 \leq 2^{O(k)} n^{3k/2} \log^{3} n.
\end{equation}
By a simple extension of \cite[Fact 3.3]{OZ13}, \eqref{eq:xor-squared} implies \eqref{eq:main-thm} in SOS :
\begin{fact}
\label{fact:sos-square-root}
\[
X^2 \leq b^2 \vdash_2 X \leq b.
\]
\end{fact}
\begin{proof}
\[
\frac{1}{2b}(b^2 - X^2) + \frac{1}{2b}(b-X)^2 = \frac{b}{2}-\frac{1}{2b}X^2 + \frac{b}{2} - X + \frac{1}{2b}X^2 = b - X. \qedhere
\]
\end{proof}

Secondly, we do not know how to prove the Cauchy-Schwarz inequality \eqref{eq:normal-cs} in SOS.   However, O'Donnell and Zhou show that a very similar inequality can be proved in SOS \cite[Fact 3.8]{OZ13}:
\begin{fact}
\label{fact:sos-cs}
\[
\vdash_2 YZ \leq \frac{1}{2}Y^2 + \frac{1}{2}Z^2.
\]
\end{fact}
Using this fact instead of Cauchy-Schwarz to prove the squared version of \eqref{eq:normal-cs}, we can follow the argument above to show that
\[
\{x_i^2 \leq 1\}_{i \in [n]} \vdash_{2k} \left(\sum_{T \in [n]^k} w(T) x^T\right)^2 \leq n\left(x^{\otimes k-1}\right)^{\top} A x^{\otimes k-1} + n\sum_{T \in [n]^k} w(T)^2.
\]
The norm bound can be proven in SOS exactly as in the even case.


\subsection{Proof of Lemma~\ref{lem:norm-bound}}
\label{sec:norm-bound-pf}
We restate the definition of the matrix $A$ and the statement of the lemma.

\begin{customlem}{\ref{lem:norm-bound}}
Let $k \geq 3$ and $p \geq n^{-k/2}$.  Let $\{w(T)\}_{T \in [n]^k}$ be indepedent random variables satisfying conditions \eqref{eq:w-mean-0}, \eqref{eq:w-mostly-0}, and \eqref{eq:w-bounded} above.  Let $A$ be the $[n]^{k-1} \times [n]^{k-1}$ indexed by $[n]^{k-1}$ that is defined as follows:
\[
A_{(i_1, i_2), (j_1,j_2)} = \begin{cases}
\sum_{\ell \in [n]} w(i_1,j_1,\ell)  w(i_2, j_2, \ell) \quad & \text{if $(i_1,j_1) \ne (i_2, j_2)$} \\
0 \quad & \quad \text{otherwise}.
\end{cases}
\]
Then with high probability,
\[
\norm{A} \leq 2^{O(k)} p n^{k/2} \log^3 n.
\]
\end{customlem}

The proof closely follows the arguments of \cite[Lemma 17]{COGL04} and \cite[Section 5]{BM15}.  Both proofs use the trace method: To bound the norm of a symmetric random matrix $M$, it suffices to bound $\E[\tr(M^r)]$ for large $r$.  For non-symmetric matrices, we can instead work with $MM^{\top}$.  In our particular case, we have the following.
\begin{claim}
If $\E[\tr((AA^{\top})^{r})] \leq n^{O(k)} 2^{O(r)} r^{6r} p^{2r} n^{kr}$, then $\norm{A} \leq 2^{O(k)} p n^{k/2} \log^3 n$ with high probability.
\end{claim}
\begin{proof}
Observe that $\|A\|^{2r} \leq \tr((AA^{\top})^{r})$.  By Markov's Inequality, $\Pr[\norm{A} \geq B] \leqÊ\frac{\E[\tr((AA^{\top})^{r})]}{B^{2r}}$.  We get the claim by plugging in $r = \Theta(\log n)$ and setting constants appropriately.
\end{proof}
\begin{remark}
\label{rem:trace-higher-prob}
We can get arbitrarily small $1/\poly(n)$ probability of failure:  This proof shows that $\norm{A} \leq K 2^{O(k)} pn^{k/2}\log^3 n$ with probability at most $n^{-\log K}$.
\end{remark}

In the the remainder of this section, we will bound $\E[\tr((AA^{\top})^{r})]$.
\begin{lemma}
Under the conditions of Lemma~\ref{lem:norm-bound}, $\E[\tr((AA^{\top})^{r})] \leq n^{O(k)} 2^{O(r)} r^{6r} p^{2r} n^{kr}$ with high probability.
\end{lemma}
\begin{proof}
Recall that we index $A$ by elements of $n^{k-1}$ divided into two blocks of $\frac{k-1}{2}$ coordinates each.  First, note that
\[
\tr((AA^{\top})^r) = \sum_{i_1,\ldots,i_{2r} \in [n]^{\frac{k-1}{2}}} (AA^{\top})_{(i_1,i_2),(i_3,i_4)}(AA^{\top})_{(i_3,i_4),(i_5,i_6)} \cdots (AA^{\top})_{(i_{2r-1},i_{2r}),(i_1,i_2)}.
\]
Expanding this out using the definition of $A$ and setting $w_T = w(T)$, we get that
\[
\tr((AA^{\top})^r) = \sum w_{i_1,j_1,\ell_1} w_{i_2,j_2,\ell_1}w_{i_3,j_1,\ell_2}w_{i_4,j_2,\ell_2} \cdots w_{i_{2r-1},j_{2r-1},\ell_{2r-1}} w_{i_{2r},j_{2r},\ell_{2r-1}} w_{i_1,j_{2r-1},\ell_{2r}} w_{i_2,j_{2r},\ell_{2r}},
\]
where the sum is over $\ell_1,\ldots,\ell_{2t} \in [n]$ and $i_1,\ldots,i_{2t},j_1,\ldots,j_{2t} \in [n]^{k-1}$ satisfying  
\begin{align}
(i_s,j_s) &\ne (i_{s+1},j_{s+1}) \quad &\text{for $1 \leq s \leq 2r-1$} \label{eq:tr1} \\
(i_{s+2}, j_s) &\ne (i_{s+3}, j_{s+1}) \quad  &\text{for $1 \leq s \leq 2r-3$} \label{eq:tr2} \\
(i_1,j_{2r-1}) &\ne (i_2,j_{2r}). \label{eq:tr3}
\end{align}
Let $\Omega$ be the set of all $(i_1,\ldots,i_{2r},j_1,\ldots,j_{2r}) \in ([n]^{\frac{k-1}{2}})^{4r}$ satisfying \eqref{eq:tr1}, \eqref{eq:tr2}, and \eqref{eq:tr3}.  Then for $J \in \Omega$ and $L = (\ell_1,\ldots,\ell_{2r}) \in [n]^{2r}$, define
\begin{equation}
\label{eq:pjl}
P_{J,L} = w_{i_1,j_1,\ell_1} w_{i_2,j_2,\ell_1}w_{i_3,j_1,\ell_2}w_{i_4,j_2,\ell_2} \cdots w_{i_{2r-1},j_{2r-1},\ell_{2r-1}} w_{i_{2r},j_{2r},\ell_{2r-1}} w_{i_1,j_{2r-1},\ell_{2r}} w_{i_2,j_{2r},\ell_{2r}}.
\end{equation}
Let $|J| = |\{i_1,\ldots,i_{2r},j_1,\ldots,j_{2r}\}|$ be the number of distinct elements of $[n]^{\frac{k-1}{2}}$ in $J$ and define $|L| = |\{\ell_1,\ldots,\ell_{2r}\}|$ similarly.
We then have
\[
\E[\tr((AA^{\top})^r)] = \sum_{J \in \Omega} \sum_{L \in [n]^{2r}} \E[P_{J,L}] = \sum_{a = 1}^{4r} \sum_{b = 1}^{2r} \sum_{\substack{J \in \Omega \\ |J| = a}} \sum_{\substack{L \in [n]^{2r} \\ |L| = b}} \E[P_{J,L}].
\]
To bound this sum, we will start by bounding $\E[P_{J,L}]$.  We will need two claims.
\begin{claim}
The number of distinct $w_{i,j,\ell}$ factors in $P_{J,L}$ is at least $2|L|$.
\end{claim}
\begin{proof}
For each $\ell \in L$, \eqref{eq:pjl} shows that $P_{J,L}$ contains a pair of the form $w_{i_s,j_s,\ell} w_{i_{s+1},j_{s+1},\ell}$ or $w_{i_{s+2} j_s,\ell}  w_{i_{s+3}, j_{s+1}, \ell}$.  Since $J \in \Omega$, we know that $(i_s,j_s) \ne (i_{s+1},j_{s+1})$ or $(i_{s+2}, j_s) \ne (i_{s+3}, j_{s+1})$, so each of these pairs must have two distinct $w_{i,j,\ell}$ factors.  We then have at least $2|L|$ distinct $w_{i,j,\ell}$ factors.
\end{proof}
\begin{claim}
The number of distinct $w_{i,j,\ell}$ factors in $P_{J,L}$ is at least $|J|-2$.
\end{claim}
\begin{proof}
Consider looking over the factors of $P_{J,L}$ from left to right in the order of \eqref{eq:pjl} until we have seen all elements of $J$.  The first pair $w_{i_1,j_1,\ell_1} w_{i_2,j_2,\ell_1}$ contains at most four previously-unseen elements of $J$.  Every subsequent pair of factors $w_{i_s,j_s,\ell_s} w_{i_{s+1},j_{s+1}}$ or $w_{i_{s+2} j_s,\ell_{s+1}}  w_{i_{s+3}, j_{s+1}, \ell_{s+1}}$ in $P_{J,L}$ shares two variables of $J$ with its preceding pair.  Each such pair can then contain at most two new elements of $J$.  After seeing $u$ $w_{i,j,\ell}$'s, we have therefore seen at most $4 + 2\left(\frac{u-2}{2}\right)$ distinct elements of $J$.  To get all $|J|$ elements of $J$, we must have seen at least $|J|-2$ $w_{i,j,\ell}$'s and these must be distinct.
\end{proof}
Since $\Pr[w_{i,j,\ell} \ne 0] \leq p$, $\E[P_{J,L}] \leq p^{\#\{\text{distinct $w_{i,j,\ell}$ factors in $P_{J,L}$}\}}$.  It then follows that
\[
\E[P_{J,L}] \leq p^{\max\{2|L|, |J|-2\}}.
\]
The two claims also imply two other facts we will need below.
\begin{claim}
If $|L| > r$, then $\E[P_{J,L}] = 0$.
\end{claim}
\begin{proof}
We will show that if $|L| > r$, there is an $w_{i,j,\ell}$ factor in $P_{J,L}$ that occurs exactly once.  Since $\E[w_{i,j,\ell}] = 0$, this proves the claim.

Assume for a contradiction that $|L| > r$ and every $w_{i,j,\ell}$ factor occurs at least twice.  Since there are at least $2|L|$ distinct $w_{i,j,\ell}$'s, there must be at least $4|L| > 4r$ total $w_{i,j,\ell}$'s.  However, looking at \eqref{eq:pjl}, $P_{J,L}$ has at most $4r$ $w_{i,j,\ell}$ factors.
\end{proof}
\begin{claim}
If $|J| > 2r+2$, then $\E[P_{J,L}] = 0$.
\end{claim}
This can be proved in exactly the same manner.

Next, observe that the number of choices of $J$ with $|J| = a$ is at most $n^{\frac{a(k-1)}{2}}a^{4r} \leq n^{\frac{a(k-1)}{2}}(4r)^{4r}$.  The number of choices of $L$ with $|L| = b$ is at most $n^b b^{2r} \leq n^b (2r)^{2r}$.  All together, we can write
\[
\E[\tr((AA^{\top})^r)] \leq \sum_{a = 1}^{2r+2} \sum_{b=1}^{r} 2^{10r} r^{6r} n^{\frac{a(k-1)}{2} + b} p^{\max\{2b, a-2\}}.
\]
We bound each term of the sum.
\begin{claim}
\[
n^{\frac{a(k-1)}{2} + b} p^{\max\{2b, a-2\}} \leq n^{kr + k -1} p^{2r}.
\]
\end{claim}
\begin{proof}
If $2b > a-2$,
\[
n^{\frac{a(k-1)}{2} + b} p^{\max\{2b, a-2\}} \leq n^{\frac{(2b+2)(k-1)}{2} + b} p^{\max\{2b, a-2\}}  = n^{k-1} (n^kp^2)^b.
\]
If $2b \leq a-2$,
\[
n^{\frac{a(k-1)}{2} + b} p^{\max\{2b, a-2\}} \leq n^{\frac{a(k-1)}{2} + \frac{a}{2} - 1} p^{a-2} = n^{k-1} (n^kp^2)^{a/2-1}.
\]
Recall that we assumed $n^kp^2 \geq 1$.  Since $a \leq 2r+2$ and $b \leq r$, the claim follows.
\end{proof}
To conclude, observe that
\[
\E[\tr[(AA^{\top})^{r}]]  \leq \sum_{a = 1}^{2r+2} \sum_{b=1}^{r} 2^{10r} r^{6r} n^{kr + k -1} p^{2r} \leq n^{O(k)} 2^{O(r)} r^{6r} p^{2r} n^{r}
\]
\end{proof}

\begin{remark}
If we did not have conditions \eqref{eq:tr1}, \eqref{eq:tr2}, and \eqref{eq:tr3}, we would only have been able to show that $|L| \leq 2r$.  This would have led to a weaker bound of $O(\sqrt{n})$.
\end{remark}


\section{Extension to larger alphabets}
\label{sec:larger-alphabets}

\subsection{Preliminaries}
\paragraph{CSPs over larger domains}
We begin by discussing CSPs over domains of size $q > 2$.  We prefer to identify such domains with~$\Z_q$, so our predicates are $P \co \Z_q^k \to \{0,1\}$.  The extensions of the definitions and facts from Section~\ref{sec:csp-prelims} are straightforward; the only slightly nonobvious notion is that of a literal.  We take the fairly standard~\cite{Aus08} definition that a literal for variable $x_i$ is any $x_i + c$ for $c \in \Z_q$. Thus there are now $q^k$ possible ``negation patterns''~$c$ for a $P$-constraint. We denote by $\calF_{q,P}(n,p)$ the distribution over instances of $\maxkcsp(P)$ in which each of the $q^kn^k$ constraints is included with probability $p$; the expected number of constraints is therefore $\expm = q^k n^k p$.
We have the following slight variant of Fact~\ref{fact:csp-facts}.
\begin{fact}
\label{fact:csp-facts-q}
Let $\calI \sim \calF_{q,P}(n,p)$.  Then the following statements hold with high probability.
    \begin{enumerate}
        \item $m = \abs{\calI} \in \expm \cdot \left(1 \pm O\left(\sqrt{\frac{\log n}{\expm}}\right)\right)$.
        \item $\val(\calI) \leq \expP \cdot \left(1 + O\left(\sqrt{\frac{\log q}{\expP} \cdot \frac{n}{\expm}}\right)\right)$.
        \item $\calI$ is $O\left(\sqrt{q^k \log q \cdot \frac{n}{\expm}}\right)$-quasirandom.
    \end{enumerate}
\end{fact}

\paragraph{Fourier analysis over larger domains} Let $\calU_q$ is the uniform distribution over $\Z_q$.  We consider the space $L^2(\Z_q,\calU_q)$ of functions $f:\Z_q \to \R$ equipped with the inner product $\la f,g \ra = \E_{\bz \sim \calU_q}[f(\bz)g(\bz)]$ and its induced norm $\norm{f}_2 = \E_{\bz \sim \calU_q}[f(\bz)^2]^{1/2}$.  Fix an orthonormal basis $\chi_0, \ldots, \chi_{q-1}$ such that $\chi_0 = 1$.

Now let $L^2(\Z_q^k, \calU_q^k)$ be the space of functions $f:\Z_q^k \to \R$, where $\calU_q^k$ is the uniform distribution over $\Z_q^k$ and we have the analogous inner product and norm.  Then, for $\sigma \in \Z_q^k$, define $\chi_{\sigma}:\Z_q^k \to \R$ such that
\[
\chi_{\sigma}(x) = \prod_{i \in [k]} \chi_{\sigma_i} (x_i).
\]
The set $\{\chi_{\sigma}\}_{\sigma \in \Z_q^k}$ forms an orthonormal basis for $L^2(\Z_q^k, \calU_q^k)$ \cite[Fact 2.3.1]{Aus08} and we can write any function $f:\Z_q^k \to \R$ in terms of this basis:
\[
f(x) = \sum_{\sigma \in \Z_q^k} \wh{f}(\sigma) \chi_{\sigma} (x).
\]
Orthonormality once again gives us Plancherel's Theorem in this setting:
\begin{theorem}
\[
\la f,g \ra = \sum_{\sigma \in \Z_q^k} \wh{f}(\sigma) \wh{g}(\sigma).
\]
\end{theorem}
For $\sigma \in \Z_q^k$, define $\supp(\sigma) = \{i \in [k]~|~\sigma_i \ne 0\}$ and $\abs{\sigma} = |\supp(\sigma)|$.  Then we define the degree of $f$ to be $\max\{\abs{\sigma}~|~\wh{f}(\sigma) \ne 0\}$.  Note that this is the degree of $f$ when it is written as a polynomial in the $\chi_a$'s for $a \in \Z_q$.  

Given a $k$-tuple $T$ and $\sigma \in \Z_q^k$, we use $T(\sigma)$ to denote the $|\sigma|$-tuple formed by taking the projection of $T$ onto the coordinates in $\supp(\sigma)$.  Similarly, use $T(\overline{\sigma})$ to denote the $(k-|\sigma|)$-tuple formed by taking the projection of $T$ onto coordinates in $[k] \setminus \supp(\sigma)$.

See \cite{OD14, Aus08} for more background on Fourier analysis over larger domains.

\subsection{Conversion to Boolean functions}
To more easily apply our above results, we would like to rewrite a function $f:\Z_q^k \to \R$ as a Boolean function $f^b:\{0,1\}^{k'} \to \R$ for some $k'$. It will actually be more convenient to define $f^b$ on a subset of $\{0,1\}^{k'}$.  In particular, consider the set $\Omega_k = \{v \in \{0,1\}^{[k] \times \Z_q}~|~\sum_{a \in \Z_q} v(i,a) = 1 ~\forall i \in [k]\}$.  Note there is a bijection $\phi$ between $\Z_q^k$ and $\Omega_k$: For $z \in \Z_q^k$, $(\phi(x))(i,a) = \indic{\{z_i = a\}}$.  In the other direction, given $v \in \Omega_k$ set $\phi^{-1}(v)_i = \sum_{a \in \Z_q} a \cdot v(i,a)$.

For a function $f:\Z_q^k \to \R$, we can then define its Boolean version $f^b: \Omega_k \to \R$ as
\[
f^b(v) = \sum_{\alpha \in \Z_q^k} f(\alpha) \prod_{i \in [k]} v(i,\alpha_i),
\]
Observe that $f(z) = f^b(\phi(z))$ for $z \in \Z_q^k$.  Also, note that if $f(z) = g(z)$ for all $z \in \Z_q^k$, $f^b = g^b$ over all $\R^k$ by construction.   $f^b$ is a multilinear polynomial and its degree is defined in the standard way.  The degree of $f$ is defined as in the previous section.
\begin{claim}
The degree of $f^b$ is equal to the degree of $f$.
\end{claim}
\begin{proof}
Abbreviate $\supp(\sigma)$ as $s(\sigma)$ and denote $\supp(\sigma)$'s complement with respect to $[k]$ as $s(\overline{\sigma})$.  Applying the definition and writing $f$'s Fourier expansion, we see that $f^b(v)$ is equal to
\[
\sum_{\alpha \in \Z_q^k} \sum_{\sigma \in \Z_q^k} \wh{f}(\sigma) \chi_{\sigma}(\alpha) \prod_{i \in [k]} v(i,\alpha_i) =  \sum_{\sigma \in \Z_q^k} \wh{f}(\sigma) \sum_{\alpha' \in \Z_q^{|\sigma|}} \chi_{\sigma}(\alpha') \prod^{|\sigma|}_{i = 1} v(s(\sigma)_i,\alpha'_i) \sum_{\alpha'' \in \Z_q^{k-|\sigma|}} \prod_{i = 1}^{k - |\sigma|} v(s(\overline{\sigma})_i,\alpha''_i).
\]
Now observe that
\[
\sum_{\alpha'' \in \Z_q^{k-|\sigma|}} \prod_{i = 1}^{k - |\sigma|} v(s(\overline{\sigma})_i,\alpha''_i) = \prod_{i = 1}^{k-|\sigma|}  \sum_{a \in \Z_q} v(s(\overline{\sigma})_i,a) = 1
\]
by the assumption that $v \in \Omega_k$.  The degree of $f^b$ is therefore $|\sigma|$.
\end{proof}

\subsection{Quasirandomness and strong refutation}
To prove quasiandomness and strong refutation results for CSPs over larger alphabets, we proceed exactly as in the binary case.  We used the $t=k$ case of Lemma~\ref{lem:AGM} (the Vazirani XOR Lemma \cite{Vaz86, Gol11}) to certify quasirandomness for binary CSPs.  A generalization of this case holds for Abelian groups \cite[Lemma 4.2]{Rao07}.
\begin{lemma}
\label{lem:xor-abelian}
Let $G$ be an Abelian group and let $\calU_G$ be the uniform distribution over $G$.  Also, let $\{\chi_{\sigma}\}_{\sigma \in G}$ be an orthonormal basis for $L^2(G, \calU_G)$ and let $D:G \to \R$ be a distribution over $G$.  If $\widehat{D}(\sigma) \leq \eps$ for all $\sigma \in G$, then $\dtv{D}{\calU_G} \leq \frac{1}{2}|G|^{3/2} \eps$.
\end{lemma}

Viewing the induced distribution density $D_{\calI,x}(\sigma)$ as a function of $x \in \Z_q^n$ for fixed $\sigma \in \Z_q^k$, we will consider $D^b_{\calI,y}(\sigma) : \Omega_n \to \R$.  As before, we can certify that $D^b_{\calI,y}$ has small Fourier coefficients.
\begin{lemma}
\label{lem:fourier-refute-q}
Let $\sigma \in \Z_q^k$ such that $\sigma \ne \mathbf{0}$ and $|\sigma| = s$.  There is an algorithm that with high probability certifies that
\[
\abs{\widehat{D^b_{\calI,y}}(\sigma)} \leq \frac{q^{O(k)} \max\{n^{s/4}, \sqrt{n}\} \log^{5/2} n}{\sqrt{\expm}}
\]
for all $y \in \{0,1\}^{[n] \times \Z_q}$ when $\expm \geq \max\{n^{s/2}, n\}$.
\end{lemma}
\begin{proof}
The proof is essentially identical to the proof of Lemma~\ref{lem:fourier-refute}.  We highlight the differences.  First of all, we can write
\[
\widehat{D^b_{\calI,y}}(\sigma) = \sum_{x \in \Z_q^n} \widehat{D_{\calI,x}}(\sigma) \prod_{i \in [n]} y(i,x_i) = \frac{1}{m} \sum_{T \in [n]^k} \sum_{c \in \Z_q^k} \indic{\{(T,c) \in \calI\}} \sum_{x \in \Z_q^n} \chi_{\sigma}(x_T + c) \prod_{i \in [n]} y(i,x_i).
\]
Since $\chi_{\sigma}$ only depends on coordinates in $\supp(\sigma)$, we can rearrange and use the fact that $\sum_{a \in \Z_q} y_{i,a} = 1$ to get
\[
\widehat{D^b_{\calI,y}}(\sigma) = \frac{1}{m} \sum_{\alpha \in \Z_q^{|\sigma|}} \sum_{U \in [n]^{|\sigma|}} \prod_{i = 1}^{|\sigma|} y(U_i,\alpha_i) \sum_{\substack{T \in [n]^k \\ T(\sigma) = U}} w_{\sigma,\alpha}(T),
\]
where $w_{\sigma,\alpha}(T) = \sum_{c \in \Z_q^{k}} \indic{\{(T,c) \in \calI\}} \chi_{\sigma}(\alpha + c(\sigma))$.   Observe that $\E[w_{\sigma, \alpha}(T)] = 0$ and $\Pr[w_{\sigma, \alpha}(T) \ne 0] \leq q^k p$.  Since $\norm{\chi_{\sigma}} = 1$, observe that the Cauchy-Schwarz Inequality implies that $|\chi_{\sigma}| \leq q^{k/2}$ for all $\sigma$.  Then $|w_{\sigma,\alpha}(T)| \leq q^{3k/2}$ for all $\alpha$ and $\sigma$.  For every $\alpha$, we can then apply Lemma~\ref{lem:smaller-monomials} just as in the proof of Lemma~\ref{lem:fourier-refute}.
\end{proof}

These two lemmas then imply the larger alphabet versions of the quasirandomness certification and strong refutation results above.
\begin{theorem}
\label{thm:quasi-q}
There is an efficient algorithm that certifies that an instance $\calI \sim \calF_{q,P}(n,p)$ of \textup{CSP}$(P)$ is $\gamma$-quasirandom with high probability when $\expm \geq \frac{q^{O(k)} n^{k/2} \log^{5} n}{\gamma^2}$.
\end{theorem}
\begin{theorem}
\label{thm:strong-ref-q}
There is an efficient algorithm that, given an instance $\calI \sim \calF_{q,P}(n,p)$ of \textup{CSP}$(P)$,  certifies that $\val(\calI) \leq \expP + \gamma$ with high probability when $\expm \geq \frac{q^{O(k)} n^{k/2} \log^{5} n}{\gamma^2}$.
\end{theorem}

\subsection{Refutation of non-$t$-wise supporting CSPs}
We will show that the dual polynomial characterization of being far from $t$-wise supporting described in Section~\ref{sec:dual-polys} generalizes to larger alphabets.  We extend the definitions of $t$-wise supporting and $\delta$-separating polynomials to the $\Z_q$ case in the natural way.

\begin{lemma}
\label{lem:poly-iff-no-dist-q}
For $P:\Z_q^k \to \{0,1\}$ and $ 0\leq \delta < 1$, there exists a polynomial $Q: \Z_q^k \rightarrow \R$ of degree at most $t$ that $\delta$-separates $P$ if and only if $P$ is $\delta$-far from supporting a $t$-wise uniform distribution.
\end{lemma}
\begin{proof}
The proof uses the following dual linear programs exactly as in the proof of Lemma~\ref{lem:poly-iff-no-dist}.
\begin{center}
\fbox{
\begin{minipage}{0.8\textwidth}
\begin{alignat}{3}
    \textrm{minimize}            &   & \sum_{\mc{z \in \Z_q^k}}(1 - &P(z))\calD(z)        &  &  \label{eq:primal-q}    \\
    \text{s.t.}\qquad & &\sum_{\mc{z\in \Z_q^k}}\calD(z) \chi_{\sigma}(z)  = q^k \wh{\calD}(\sigma)           & =            0                  && \forall \sigma \in \Z_q^k  \quad0 < |\sigma| \leq t  \label{eq:0-bias-q} \\
                      &   &       \sum_{{z  \in \Z_q^k}}\calD(z)             & =            1                  &   &      \nonumber \\
                      &   & \calD(z)         & \geq            0    &&\forall z \in \Z_q^k \nonumber
 \end{alignat}
\end{minipage}
}
\end{center}
\begin{center}\fbox{
\begin{minipage}{0.6\textwidth}
\begin{alignat}{3}
    \textrm{maximize}            &   & \xi &   &  &    \nonumber  \\
    \text{s.t.}\qquad                 &   &       \sum_{\substack{\sigma \in \Z_q^k\\0 < |\sigma| \leq t}}c(S) \chi_{\sigma}(z)   &\leq 1-P(z) -\zeta  &\qquad  &\forall z \in \Z_q^k.      \nonumber
     \end{alignat}
\end{minipage}}
\end{center}

To prove Lemma~\ref{lem:poly-iff-no-dist}, we needed to show in the binary case that feasible solutions to the primal LP \eqref{eq:primal-lp} were $t$-wise uniform.  We now argue that the constraint \eqref{eq:0-bias-q} is a sufficient condition for $t$-wise uniformity of $\calD$ in the $q$-ary case.  For a distribution $\calD$ over $\Z_q^k$ and $S \subseteq [k]$, define $\calD_S$ to be the marginal distribution of $\calD$ on $(\Z_q^k)_S$, i.e., $\calD_S(z) = \sum_{z' \in \Z_q^k, z'_S = z} \calD(z')$.  We need to show that \eqref{eq:0-bias-q} implies that $\calD_S = \calU^{|S|}_q$ for all $S \subseteq [k]$ with $1 \leq |S| \leq t$.

Fix such an $S$ and let $|S| = s$.  Consider the basis $\{\chi_{\alpha}\}_{\alpha \in \Z_q^s}$.  Lemma~\ref{lem:xor-abelian} implies that it suffices to show that $\E_{\bz \sim \calU_q^s}[D_S(\bz)\chi_{\alpha}(\bz)] = 0$ for all $\alpha \in \Z_q^s$.  Observe that $\E_{\bz \sim \calU_q^s}[D_S(\bz)\chi_{\alpha}(\bz)] = \E_{\bz' \sim \calU_q^k}[\calD(\bz') \chi_{\sigma}(\bz')]$ for $\sigma \in \Z_q^k$ such that $\sigma_i = \alpha_i$ for $i \in S$ and $\sigma_i = 0$ otherwise.  Since $|S| \leq t$, we know that $|\sigma| \leq t$ and \eqref{eq:0-bias-q} implies $ \E_{\bz' \sim \calU_q^k}[\calD(\bz') \chi_{\sigma}(\bz')] =0$.

The rest of the proof is exactly as in the binary case.
\end{proof}

We can again use these separating polynomials to obtain almost $\delta$-refutation for predicates that are $\delta$-far from $t$-wise supporting.
\begin{theorem}
\label{thm:ntwus-refute-q}
Let $P$ be $\delta$-far from being \twus.  There exists an efficient algorithm that, given an instance $\calI \sim \calF_{q,P}(n,p)$ of $\CSP(P)$, certifies that $\val(\calI) \leq 1 - \delta + \gamma$ with high probability when $\expm \geq \frac{q^{O(k)} n^{t/2} \log^{5} n}{\gamma^2}$ and $t \geq 2$.
\end{theorem}
The proof is essentially identical to Proof~2 of Theorem~\ref{thm:ntwus-refute-2}.

Corollary~\ref{cor:no-t-wise-ref-2} also extends to larger alphabets.
\begin{corollary}
\label{cor:no-t-wise-ref-q}
Let $P$ be a predicate that does not support any $t$-wise uniform distribution.  Then there is an efficient algorithm that, given an instance $\calI \sim \calF_{q,P}(n,p)$ of $\CSP(P)$, certifies that $\val(\calI) \leq 1 - 2^{-\wt{O}(q^t k^t)}$ with high probability when $\expm \geq 2^{\wt{O}(q^t k^t)} n^{t/2} \log^{5} n$ and $t \geq 2$.
\end{corollary}
This follows directly from Theorem~\ref{thm:ntwus-refute-q} and the following extension of Corollary~\ref{cor:lp-granularity} to larger alphabets.

\begin{corollary}                                       \label{cor:lp-granularity-q}
    Suppose $P \co \Z_q^k \to \{0,1\}$ is not \twus. Then it is in fact $\delta$-far from \twus for $\delta = 2^{-\wt{O}(q^t k^t)}$.
\end{corollary}
The proof is essentially identical to the proof of Corollary~\ref{cor:lp-granularity}: Observe that the LP \eqref{eq:primal-q} has at most $q^tk^t$ variables and proceed exactly as before.

\subsection{SOS proofs}

Here we give SOS versions of our refutation results for larger alphabets.

\paragraph{Certifying Fourier coefficients are small}
To give an SOS proof that Fourier coefficients of $D^b_{\calI,y}$ are small, we again need to define a specific polynomial representation of $\widehat{D^b_{\calI,y}}(\sigma)$.
\[
\widehat{D_{\calI,y}}(\sigma)^{\poly} = \frac{1}{m} \sum_{T \in [n]^k} \sum_{c \in \Z_q^k} \indic{\{(T,c) \in \calI\}}  \sum_{\alpha \in \Z_q^{|\sigma|}} \chi_{\sigma}(\alpha + c(\sigma)) \prod_{i = 1}^{|\sigma|} y(T(\sigma)_i,\alpha_i).
\]

\begin{lemma}
\label{lem:fourier-refute-sos-q}
Let $\mathbf{0} \ne \sigma \in \Z_q^k$ with $|\sigma| = s$.  Then
\begin{align*}
\{y(i,a)^2 \leq 1\}_{\substack{i \in [n]\\ a \in \Z_q}} &\vdash_{\max\{2s,k\}} \widehat{D_{\calI,y}}(\sigma)^{\poly} \leq \frac{q^{O(k)} \max\{n^{s/4}, \sqrt{n}\} \log^{5/2} n}{\expm^{1/2}} \\
\{y(i,a)^2 \leq 1\}_{\substack{i \in [n]\\ a \in \Z_q}} &\vdash_{\max\{2s,k\}} \widehat{D_{\calI,y}}(\sigma)^{\poly} \geq -\frac{q^{O(k)} \max\{n^{s/4}, \sqrt{n}\} \log^{5/2} n}{\expm^{1/2}}.
\end{align*}
with high probability, assuming also that $\expm \geq \max\{n^{s/2}, n\}$.
\end{lemma}
\begin{proof}
In the proof of Lemma~\ref{lem:fourier-refute-q}, we certify that $|\widehat{D_{\calI,y}}(\sigma)|$ is small by certifying that $\abs{\widehat{D_{\calI,y}}(\sigma)^{\poly}}$ is small.  The proof of Lemma~\ref{lem:fourier-refute-q} relies only on Lemma~\ref{lem:smaller-monomials}; we can replace this with its SOS version Lemma~\ref{lem:smaller-monomials-sos}.
\end{proof}
\begin{remark}
We stated the lemma with the weaker set of axioms $\{y(i,a)^2 \leq 1\}_{i \in [n], a \in \Z_q}$.  Since $y(i,a)^2 = y(i,a)$ implies $y(i,a)^2 \leq 1$ in degree-2 SOS, the lemma holds with the axioms $\{y(i,a)^2 = y(i,a)\}_{i \in [n], a \in \Z_q}$ as well.
\end{remark}

\paragraph{Strong refutation of any $k$-CSP}
From our SOS proof that the Fourier coefficients $\widehat{D^b_{\calI,y}}(\sigma)$ are small, we can get SOS proofs of strong refutation for any $k$-CSP.  To do this, we need to define a specific polynomial representation of $\mathrm{Val}^b_{\calI}(y)$ for an instance $\calI$ of CSP$(P)$:
\[
\mathrm{Val}_{\calI}(y)^{\poly} = \frac{1}{m} \sum_{T \in [n]^k} \sum_{c \in \Z_q^k} \indic{\{(T,c) \in \calI\}} \sum_{\alpha \in \Z_q^k} P(\alpha + c) \prod_{i \in [k]} y(T_i,\alpha_i).
\]
\begin{theorem} \label{thm:strong-ref-q-SOS}
Given an instance $\calI \sim \calF_{q,P}(n,p)$ of \textup{CSP}$(P)$,
\[
\{y(i,a)^2 = y(i,a)\}_{\substack{i \in [n]\\ a \in \Z_q}} \cup \left\{\sum_{a \in \Z_q} y(i,a) = 1\right\}_{i \in [n]} \vdash_{2k} \mathrm{Val}_{\calI}(y)^{\poly} \leq \expP + \gamma
\]
with high probability when $\expm \geq \frac{q^{O(k)} n^{k/2} \log^{5} n}{\gamma^2}$.
\end{theorem}
\begin{proof}
First, use the Fourier expansion of $P$ to write
\[
\mathrm{Val}_{\calI}(y)^{\poly} = \frac{1}{m} \sum_{T \in [n]^k} \sum_{c \in \Z_q^k} \indic{\{(T,c) \in \calI\}} \sum_{\alpha \in \Z_q^k} \sum_{\sigma \in \Z_q^k} \wh{P}(\sigma) \chi_{\sigma}(\alpha + c) \prod_{i \in [k]} y(T_i,\alpha_i).
\]
For each $T \in [n]^k$, $c \in \Z_q^k$, and $\sigma \in \Z_q^k$, we have a term of the form $\sum_{\alpha \in \Z_q^k}\chi_{\sigma}(\alpha + c) \prod_{i \in [k]} y(T_i,\alpha_i)$.

Note that $\chi_{\sigma}$ only depends on the coordinates in $\supp(\sigma)$.  We can then write this as
\[
\left(\sum_{\alpha \in \Z_q^{|\sigma|}} \chi_{\sigma}(\alpha+c(\sigma)) \prod_{i=1}^{|\sigma|}y(T(\sigma)_i, \alpha_i)\right)\left(\prod_{i = 1}^{k-|\sigma|} \sum_{a \in \Z_q} y(T(\overline{\sigma})_i, a)\right)
\]
Using the axioms $\sum_{a \in \Z_q} y_{i,a} = 1$, the second term is equal to $1$ and we have
\[
\left\{\sum_{a \in \Z_q} y(i,a) = 1\right\}_{i \in [n]} \vdash_k \sum_{\alpha \in \Z_q^k}\chi_{\sigma}(\alpha + c) \prod_{i \in [k]} y(T_i,\alpha_i) =  \sum_{\alpha \in \Z_q^{|\sigma|}} \chi_{\sigma}(\alpha+c(\sigma)) \prod_{i=1}^{|\sigma|}y(T(\sigma)_i, \alpha_i).
\]
Summing over all $T$, $c$, and $\sigma$, we obtain the following.
\[
\left\{\sum_{a \in \Z_q} y(i,a) = 1\right\}_{i \in [n]} \vdash_k \mathrm{Val}_{\calI}(y)^{\poly} = \frac{1}{m} \sum_{T \in [n]^k} \sum_{c \in \Z_q^k} \indic{\{(T,c) \in \calI\}} \sum_{\sigma \in \Z_q^k} \wh{P}(\sigma) \sum_{\alpha \in \Z_q^{|\sigma|}} \chi_{\sigma}(\alpha+c(\sigma)) \prod_{i=1}^{|\sigma|}y(T(\sigma)_i, \alpha_i).
\]
This is equal to
\[
\expP + \sum_{\mathbf{0} \ne \sigma \in \Z_q^k} \wh{P}(\sigma) \widehat{D_{\calI,y}}(\sigma)^{\poly}.
\]
Since $|P(z)| \leq 1$ and $|\chi_{\sigma}(z)| \leq q^{O(k)}$, $\sum_{\sigma \in \Z_q^k} |\wh{P}(\sigma)| \leq q^{O(k)}$.  We can then apply Lemma~\ref{lem:fourier-refute-sos-q} for each $\sigma$ to complete the proof.
\end{proof}

\paragraph{SOS refutation of non-$t$-wise supporting CSPs}
\begin{theorem} \label{thm:ntwus-refute-q-sos}
Let $P$ be $\delta$-far from being \twus.  Then, given an instance $\calI \sim \calF_{q,P}(n,p)$ of $\CSP(P)$,
\[
\left\{y(i,a)^2 = y(i,a)\right\}_{\substack{i \in [n]\\ a \in \Z_q}} \cup \left\{y(i,a) y(i,b) = 0\right\}_{\substack{i \in [n]\\ a \ne b \in \Z_q}}  \cup \left\{\sum_{a \in \Z_q} y(i,a) = 1\right\}_{i \in [n]} \vdash_{\max\{k,2t\}} \mathrm{Val}_{\calI}(y)^{\poly} \leq 1-\delta + \gamma.
\]
with high probability when $\expm \geq \frac{q^{O(k)} n^{t/2} \log^{5} n}{\gamma^2}$ and $t \geq 2$.
\end{theorem}
To prove this theorem, we need a version of Claim~\ref{cl:deg-k-ineq} for larger alphabets.
\begin{claim}
\label{cl:deg-k-ineq-q}
Let $f:\Z_q^k \to \R$ such that $f(z) \geq 0$ for all $z \in \Z_q^k$ and let $f^b(v) = \sum_{\alpha \in \Z_q^k} f(\alpha) \prod_{i \in [k]} v(i,\alpha_i)$.  Then
\[
\left\{v(i,a)^2 = v(i,a)\right\}_{\substack{i \in [k]\\ a \in \Z_q}} \cup \left\{v_{i,a} v_{i,b} = 0\right\}_{\substack{i \in [k]\\ a \ne b \in \Z_q}} \vdash_k f^b(v) \geq 0.
\]
\end{claim}
\begin{proof}
Since $f(z) \geq 0$ for all $z \in \Z_q^k$, there exists a function $g:\Z_q^k \to \R$ such that $g^2(z) = f(z)$ for all $z \in \Z_q^k$.  We then write $g^b(v) = \sum_{\alpha \in \Z_q^k} g(\alpha) \prod_{i \in [k]} v(i,\alpha_i)$.  Using $v(i,a)^2 = v(i,a)$, it follows that
\[
g^b(v)^2 = \sum_{\alpha \in \Z_q^k} g(\alpha)^2 \prod_{i \in [k]} v(i,\alpha_i) + \sum_{\alpha' \ne \alpha'' \in \Z_q^k} g(\alpha')g(\alpha'') \prod_{i \in [k]} v(i,\alpha'_i) v(i,\alpha''_i)
\]
The first term is equal to $f^b(\sigma)$.  For the second term, note that each of the products $\prod_{i \in [k]} v(i,\alpha'_i) v(i,\alpha''_i)$ must contain factors $v(i,a) v(i,b)$ with $a \ne b$ since $\alpha' \ne \alpha''$.  We have the axiom $v(i,a) v(i,b) = 0$, so the second term is $0$.  Then $f^b = (g^b)^2$ and the claim follows.
\end{proof}

With this claim, the proof of the theorem exactly follows that of Theorem~\ref{thm:ntwus-refute-2}.
\begin{proof}[Proof of Theorem~\ref{thm:ntwus-refute-q-sos}]
Claim~\ref{cl:deg-k-ineq-q} implies that
\[
\left\{v(i,a)^2 = v(i,a)\right\}_{\substack{i \in [k]\\ a \in \Z_q}} \cup \left\{v(i,a) v(i,b) = 0\right\}_{\substack{i \in [k]\\ a \ne b \in \Z_q}}  \cup \left\{\sum_{a \in \Z_q} v(i,a) = 1\right\}_{i \in [k]} \vdash_{k} P^b(v) - (1-\delta) \leq Q^b(v)
\]
Summing over all constraints, we get that
\[
A \vdash_{k} m\mathrm{Val}_{\calI}(y)^{\poly} - m(1-\delta) \leq \sum_{T \in [n]^k} \sum_{c \in \Z_q^k} \indic{\{(T,c) \in \calI\}} \sum_{\alpha \in \Z_q^k} Q(\alpha+c) \prod_{i = 1}^k y(i,\alpha_i)
\]
where $A = \left\{y(i,a)^2 = y(i,a)\right\}_{\substack{i \in [n]\\ a \in \Z_q}} \cup \left\{y(i,a) y(i,b) = 0\right\}_{\substack{i \in [n]\\ a \ne b \in \Z_q}}  \cup \left\{\sum_{a \in \Z_q} y(i,a) = 1\right\}_{i \in [n]}$.
Using the Fourier expansion of $Q$, we see that the right-hand side of the inequality is
\[
\sum_{T \in [n]^k} \sum_{c \in \Z_q^k} \indic{\{(T,c) \in \calI\}} \sum_{\alpha \in \Z_q^k} \sum_{\sigma \in \Z_q^k} \wh{Q}(\sigma)\chi_{\sigma}(\alpha+c) \prod_{i = 1}^k y(T_i,\alpha_i)
\]
Just as in the proof of Theorem~\ref{thm:strong-ref-q-SOS}, we can rewrite this in degree-$k$ SOS as
\[
\sum_{T \in [n]^k} \sum_{c \in \Z_q^k}  \indic{\{(T,c) \in \calI\}} \sum_{\sigma \in \Z_q^k} \wh{Q}(\sigma) \sum_{\alpha \in \Z_q^{|\sigma|}} \chi_{\sigma}(\alpha+c(\sigma)) \prod_{i = 1}^{|\sigma|} y(T(\sigma)_i,\alpha_i).
\]
We then rearrange to get
\[
\sum_{\mathbf{0} \ne \sigma \in \Z_q^k} \wh{Q}(\sigma) \widehat{D_{\calI,y}}(\sigma)^{\poly}.
\]
Since $\E[Q] = 0$ and $Q \geq -1$, we know that $|Q| \leq q^{O(k)}$ and therefore $|\wh{Q}(\sigma)| \leq q^{O(k)}$. We can then apply Lemma~\ref{lem:fourier-refute-sos-q} for each $\sigma$ to complete the proof.
\end{proof}

\section{Certifying that random hypergraphs have small independence number and large chromatic number} \label{sec:hypergraphs}
First, we recall some standard definitions.  Let $H = (V,E)$ be a hypergraph.  We say that $S$ is an independent set of $H$ if for all $e \in E$, it holds that $e \notin S$.  The independence number $\alpha(H)$ is then the size of the largest independent set of $H$.  A $q$-coloring of $H$ is a function $f:V \to [q]$ such that $f^{-1}(i)$ is an independent set for every $i \in [q]$.  The chromatic number $\chi(H)$ is the the smallest $q \in \N$ for which there exists a $q$-coloring of $H$.

We define $\calH(n,p,k)$ to be the distribution over $n$-vertex, $k$-uniform (unordered) hypergraphs in which each of the $\binom{n}{k}$ possible hyperedges is included independently with probability $p$.  Let $\expm$ be the expected number of hyperedges $p \binom{n}{k}$.

Coja-Oghlan, Goerdt, and Lanka used CSP refutation techniques to show the following results \cite{COGL07}:
\begin{theorem} \textup{(Coja-Oghlan--Goerdt--Lanka~\cite[Theorem~3]{COGL07}).}
For $H \sim \calH(n,p,3)$, there is a polynomial time algorithm certifying that $\alpha(H) < \eps n$ with high probability for any constant $\eps > 0$ when $\expm > n^{3/2}\ln^6 n$ and $\expm = o(n^2)$.
\end{theorem}

\begin{theorem} \textup{(Coja-Oghlan--Goerdt--Lanka~\cite[implicit in Section~4]{COGL07}).}
For $H \sim \calH(n,p,4)$, there is a polynomial time algorithm certifying that $\alpha(H) < \eps n$ with high probability for any constant $\eps > 0$ when $\expm \geq O\left(\frac{n^2}{\eps^4}\right)$.
\end{theorem}

\begin{theorem} \textup{(Coja-Oghlan--Goerdt--Lanka~\cite[Theorem~4]{COGL07}).}
For $H \sim \calH(n,p,4)$, there is a polynomial time algorithm certifying that $\chi(H) > \xi$ with high probability for constant $\xi$ when $\expm \geq O(\xi^4 n^2)$.
\end{theorem}

We generalize these results to $k$-uniform hypergraphs:
\begin{theorem} \label{thm:k-ind-set}
For $H \sim \calH(n,p,k)$, there is a polynomial time algorithm certifying that $\alpha(H) < \beta$ with high probability when $\expm \geq O_{k}\left(\frac{n^{5k/2} \log^3 n}{\beta^{2k}}\right)$, assuming that $\beta \geq n^{3/4} \log n$.
\end{theorem}

\begin{theorem} \label{thm:k-chrom-num}
For $H \sim \calH(n,p,k)$, there is a polynomial time algorithm certifying that $\chi(H) > \xi$ with high probability when $\expm \geq O_{k}\left(\xi^{2k} n^{k/2} \log^3 n\right)$, assuming that $\xi \leq \frac{n^{1/4}}{\log n}$.
\end{theorem}

The proofs are simple extensions of the $k = 3$ and $k = 4$ cases from \cite{COGL07}.  We will first prove Theorem~\ref{thm:k-ind-set} using Theorem~\ref{thm:main} and this will almost immediately imply Theorem~\ref{thm:k-chrom-num}.

\begin{proof}[Proof of Theorem~\ref{thm:k-ind-set}]
Recall that Theorem~\ref{thm:main} deals with $k$-tuples, not sets of size $k$.  It is easy to express a hypergraph in terms of $k$-tuples rather than sets of size $k$.  For a set $S$ and $t \in \Z_{\geq 0}$, recall the notation $\binom{S}{t} = \{T \subseteq S~|~|T| = t\}$.  For each possible hyperedge $e \in \binom{[n]}{k}$, we associate an arbitrary tuple $T_e$ from among the $k!$ tuples in $[n]^k$ containing the same $k$ elements.  To draw from $\calH(n,p,k)$, we include each $T_e$ independently with probability $p$ and include all other $T \in [n]^k$ with probability $0$.

For $T \in [n]^k$, we define the random variable $w(T)$ as follows:
\[
w(T) = \begin{cases}
p - \indic{\{e \in E\}} & \text{if $T = T_e$ for some $e \in \binom{[n]}{k}$} \\
0 & \text{otherwise.}
\end{cases}
\]
Let $x \in \{0,1\}^n$ be the indicator vector of an independent set $I$ so that $x^T = 1$ if $T \subseteq I$ and $x^T = 0$ otherwise.  First, observe that
\[
\sum_{T \in [n]^k} w(T) x^T = p\sum_{S \in \binom{[n]}{k}} x^S - \sum_{e \in \binom{[n]}{k}} \indic{\{e \in E\}} x^{e} = p \binom{|I|}{k},
\]
where the second term is $0$ because $I$ is an independent set.  The $w(T)$'s satisfy conditions \eqref{eq:w-mean-0}, \eqref{eq:w-mostly-0}, and \eqref{eq:w-bounded} and $\|x\|_{\infty} \leq 1$, so Theorem~\ref{thm:main} implies we can certify that
\[
p \binom{|I|}{k} = \sum_{T \in [n]^k} w(T) x^T \leq 2^{O(k)} \sqrt{p} n^{3k/4} \log^{3/2} n
\]
with high probability.  Simplifying, we see that we can certify
\[
|I| \leq O_k\left(\frac{n^{5/4} \log^{\frac{3}{2k}} n}{\expm^{\frac{1}{2k}}}\right)
\]
and plugging in the value of $\expm$ from the statement of the theorem completes the proof.
\end{proof}

\begin{proof}[Proof of Theorem~\ref{thm:k-chrom-num}]
For a coloring of a hypergraph $H$, each color class is an independent set of $H$.  If $\chi(H) \leq \xi$, then there exists a color class of size at least $\frac{n}{\xi}$ and therefore $\alpha(H) \geq \frac{n}{\xi}$.  We can then certify that $\alpha(H) < \frac{n}{\xi}$ using Theorem~\ref{thm:k-ind-set}.
\end{proof}


\section{Simulating $\calF_P(n, p)$ with  a fixed number of constraints } \label{sec:translation}
The setting of~\cite{daniely-hardness-of-learning}  fixes the number of constraints in a CSP instance, whereas the model described in Section~\ref{sec:prelims} includes each possible constraint in the instance with some probability $p$.
Here we show that results from our setting easily extend to that of~\cite{daniely-hardness-of-learning} by giving an algorithm that simulates the behavior of our model when the number of constraints is fixed.

Recall that an instance $\calI \sim \calF_P(n, p)$ is generated as follows.  For each $S \in [n]^k$ and each $c \in \{-1, 1 \}^k$, constraint $(c, S)$ is included with probability $p$, so the expected number of constraints is $p \cdot (2n)^k$.

In the model where the number of constraints is fixed, the instance is guaranteed to have $m$ distinct constraints for some value of $m$.  The instance $\calJ$ is chosen uniformly from all subsets of $\{-1, 1\}^k \times [n]^k$ with size exactly $m$ 
\begin{theorem}\label{thm:translate-models.}
Suppose there exists an efficient algorithm $R$ that, on a given CSP instance $\calI \sim \calF_P(n, p)$, for all $p \geq p_{\min}$, certifies that $\val(\calI) \leq \eta$ for some $0\leq \eta < 1$ with high probability.
Then there exists an efficient algorithm $\calA$ that certifies that a random instance $\calJ$ of $\CSP(P)$ with $\mu$ constraints has $\val(\calJ) \leq \eta + 2\left(\mu^{-1}\ln\mu\right)^{1/2}$ with high probability when $\mu\left(1 - \left(\mu^{-1}\ln\mu\right)^{1/2}\right) \geq (2n)^kp_{\min}$.
 \end{theorem}
\begin{proof}
On a random instance with $\mu$ constraints, we can generate an instance $\calI$ that simulates this behavior by choosing an appropriate value for $p$, drawing $m\sim\textrm{Binomial}\left(p, (2n)^k\right)$ and then discarding $\mu - m$ of the constraints.  For brevity, let $d = \left(\mu^{-1}\ln\mu\right)^{1/2}.$
Algorithm~\ref{alg:refute} describes the behavior of $\calA$.
\begin{algorithm}{Algorithm $\calA$}\
\begin{algorithmic}[1]
\State $p \gets \mu(1-d)(2n)^{-k}$.
\State draw $m \sim \textrm{Binomial}\left(p, (2n)^k \right)$
\If {$m > \mu$ or   $m < \mu(1 - 2d)$}
\State \Return ``fail."
\EndIf
\State $\calI \gets \calJ$
\For{$i = m + 1 \ldots  \mu$}
\State Remove a random constraint from $\calI$ chosen uniformly
\EndFor
\State Run $R$ on $\calI$
\If{$R$ certifies that $\val(\calI) \leq \eta$}
\State \Return ``$\val(\calJ) \leq \eta + 2d$."
\Else
\State \Return ``fail."
\EndIf
\end{algorithmic}
\caption{}
 \label{alg:refute}
\end{algorithm}

The fraction of removed constraints is at most $2d$, so even if all of the removed constraints would have been satisfied, their contribution to $\val(\calJ)$ is at most $2d$.
Consequently, $\calA$ will never incorrectly output ``$\val(\calJ) \leq  \eta + 2d$."

Furthermore, the probability of failing to refute an instance with value at most $ 1 - \eta + 2d$ due to exiting at step 2 is $o_{k, t}(1)$.
We treat $m$ as a sum of $(2n)^k$ independent Bernoulli variables with probability $p$ and denote $\E[m]$ by $\expm$.  Applying a Chernoff bound yields the following.
\begin{align*}
\Pr[ m > \mu] &= \Pr[m > \expm/(1 - d)]\\
&=\Pr[m > \expm(1 + \tfrac{d}{1-d})]\\
&\leq \textrm{exp}\left(-\tfrac{\mu^{-1}\mu\ln\mu(1-(\mu\ln\mu)^{-1/2})}{3}\right)\\
&< \textrm{exp}\left(-\Theta(\ln\mu)\right) = 1/{\poly(\mu)}.
\end{align*}

Similarly,
\begin{align*}
\Pr[m < \mu(1-2d)] &= \Pr[m < \expm(1-\tfrac{d}{1-d})]\\
&\leq \textrm{exp}\left(-\Theta(\ln \mu)\right) = 1/{\poly(\mu)}.
\end{align*}
If $\mu(1 - \left(\mu^{-1}\ln \mu\right)^{1/2}) \geq (2n)^kp_{\min}$, then $p \geq p_{\min}$ and $R$ will be able to certify $\val(\calI) \leq \eta$ with high probability.
\end{proof}

\end{document}